%% file: CC-LS1.tex
\documentclass[11pt,letterpaper]{article}
\usepackage{amsmath}
\usepackage{amsfonts}
\usepackage{amssymb}
\usepackage{graphicx}
\usepackage{amsmath, amssymb, amsthm}
\usepackage[left=1.00in, right=1.00in, top=1.00in, bottom=1.00in]{geometry}

\usepackage{amsfonts,amsthm,amssymb}
\usepackage{amssymb}
\usepackage{amsmath}
\usepackage{booktabs}
\usepackage{appendix}
\usepackage{etoolbox}
\usepackage{natbib}
\usepackage{MnSymbol}
\usepackage{tikz}
\usetikzlibrary{backgrounds, shapes, snakes, matrix}
\usetikzlibrary{backgrounds, matrix, positioning}
\usepackage{dsfont}

\newcommand{\R}{\mathbb{R}}

\newcommand{\QLS}{\textsc{QuLS}}
\newcommand{\SimLS}{\textsc{SimLS}}
\newcommand{\SumLS}{\textsc{SumLS}}
\newcommand{\ValLS}{\textsc{VetoLS}}
\newcommand{\ra}{\rightarrow}
\newcommand{\0}{\textbf{0}}
\newcommand{\1}{\mathds{1}}
\newcommand{\Hyp}{\textsf{Hyp}}
\newcommand{\Odd}{\textsf{Odd}}
\newcommand{\Grid}{\textsf{Grid}}
\newcommand{\opi}{\overline{\pi}}
\newcommand{\upi}{\underline{\pi}}

\newcommand{\poly}{\textsf{poly}}
\newcommand{\polylog}{\textsf{polylog}}

\newtheorem{theorem}{Theorem}[section]
\newtheorem{theo}{Theorem}
\newtheorem*{maintheo*}{Main Theorem}
\newtheorem{corollary}[theorem]{Corollary}

\newtheorem{lemma}[theorem]{Lemma}
\newtheorem{proposition}[theorem]{Proposition}

\newtheorem{definition}[theorem]{Definition}

\begin{document}

\title{The Communication Complexity of Local Search}
\author{Yakov Babichenko \and Shahar Dobzinski \and Noam Nisan}

\date{}

\maketitle

\begin{abstract}
We study the following communication variant of local search. There is some fixed, commonly known graph $G$. Alice holds $f_A$ and Bob holds $f_B$, both are functions that specify a value for each vertex. The goal is to find a local maximum of $f_A+f_B$ with respect to $G$, i.e., a vertex $v$ for which $(f_A+f_B)(v)\geq (f_A+f_B)(u)$ for every neighbor $u$ of $v$.

Our main result is that finding a local maximum requires polynomial (in the number of vertices) bits of communication. The result holds for the following families of graphs: three dimensional grids, hypercubes, odd graphs, and degree 4 graphs. Moreover, we provide an \emph{optimal} communication bound of $\Omega(\sqrt{N})$ for the hypercube, and for a constant dimensional greed, where $N$ is the number of vertices in the graph.

We provide applications of our main result in two domains, exact potential games and combinatorial auctions. First, we show that finding a pure Nash equilibrium in $2$-player $N$-action exact potential games requires polynomial (in $N$) communication. We also show that finding a pure Nash equilibrium in $n$-player $2$-action exact potential games requires exponential (in $n$) communication.

The second domain that we consider is combinatorial auctions, in which we prove that finding a local maximum in combinatorial auctions requires exponential (in the number of items) communication even when the valuations are submodular.

Each one of the results demonstrates an exponential separation between the non-deterministic communication complexity and the randomized communication complexity of a total search problem.
\end{abstract}

\thispagestyle{empty}
\newpage
\setcounter{page}{1}

\input{intro}

\input{graphs}

\input{potential2n}

\paragraph{Acknowledgement.} We are very grateful to Mika G{\"o}{\"o}s and Aviad Rubinstein for drawing our attention to pebbling games \cite{GP14}, which allowed us to prove \emph{optimal} communication bounds for several families of graphs.

\bibliography{CCPLS}

\appendix

\input{id-ord}

\end{document}

%% file: intro.tex
%
%
%
%

\section{Introduction}

This paper deals with the communication complexity of local search problems.  The general problem involves a search over some ``universe'' $V$, for an element $v^* \in V$ that maximizes, {\em at least ``locally''}, some objective function $f:V\rightarrow \R$.  The notion of ``locality'' is formalized by putting a fixed, known, neighbourhood structure $E$ on the set of elements, so the requirement of local optimality is that for all $u \in V$ such that $(v^*,u) \in E$ we have that $f(v^*) \ge f(u)$.  The notion of local optimality is interesting from two points of view: first, it captures the outcome of a wide range of ``gradual-improvement'' heuristics where the neighbourhood structure represents the types of gradual improvements allowed, and second, locally-optimal solutions provide a notion of stability, where the neighborhood structure models the possible ``deviations'' from stability.  

In the context of computational complexity, local search problems are captured by the complexity class PLS \citep{JPY} which is a subset of the well studied class TFNP (defined in \citep{MP} and studied, e.g., in \citep{PSY, BCEI, DGP, hubacek2017journey}): search problems for which a witness always exists (``total search  problems'') and can by efficiently verified (``in NP'').  The problem has also been widely studied in the model of {\em query complexity} where the cost of an algorithm is the number of black-box queries to the objective function $f$, from the pioneering work of \citep{Ald} on the Boolean hypercube, to a rather complete characterization of not only the deterministic query complexity but also the randomized and even quantum complexities on any graph \citep{SS,Aar,SY}. 

The interest in analyzing local search from a communication complexity point of view is clear: in essentially any application, the objective function $f$ is not really given as a ``black box'' but is somehow determined by the problem structure. When this structure has any element of distributed content then communication may become an important bottleneck. The question of how the information is distributed is key: in the simplest imaginable scenario, the search space $V$ is split somehow between the (say, two) parties, where each party holds the values $f(v)$ for its subset of $v \in V$ (the fixed commonly known neighbourhood structure still involves all of $V$).  However, in this scenario even a global maximum (which is certainly also a local one) can be easily found with a small amount of communication by each player finding the maximum among his subset, and only communicating and comparing the maxima of the parties. Thus, for the problem to be interesting we must split the information $f(v)$ of each vertex between the parties. There are various ways to do this and the most natural one, conceptually and in terms of applications, is probably to split $f$ as the {\em sum} of two functions $f_A:V\rightarrow \R$ and $f_B:V\rightarrow \R$ held by Alice and Bob. So we consider the following problem:

\vspace{0.1in}
\noindent
{\bf Definition:} For a fixed, commonly known graph $G=(V,E)$, the $\SumLS(G)$ communication problem is the following:
Alice holds a function $f_A:V\rightarrow \{1,...,W\}$,
Bob holds a function $f_B:V\rightarrow \{1,...,W\}$, and their goal is to find a vertex $v^* \in V$ such that $f_A(v^*)+f_B(v^*) \ge f_A(u)+f_B(u)$ for all $u \in V$ with $(v^*,u) \in E$.
\vspace{0.1in}

\noindent Determining the communication complexity of $\SumLS$ on certain families of graphs is easy. For example, a simple reduction from disjointness shows that the communication complexity of $\SumLS$ on the clique with $n$ vertices is $\Omega(n)$. Our main theorem proves optimal lower bounds for several important families of graphs, all have small degree. The technical challenge is that the non-deterministic communication complexity of the problem on small degree graphs is clearly low: to {\em verify} that $v^*$ is a local optimum, Alice and Bob need only communicate the values $f(u)$ and $g(u)$ for the small number of $v^*$'s neighbours in the graph (note that the degree of all graphs that we consider is indeed small: $\log N$ or even constant).  There are only a few results in the communication complexity literature that manage to prove good lower bounds for total problems where verification is easy, most notably for Karchmer-Wigderson games \citep{KW,KRW,RM} and for PPAD-like communication problems \citep{BR,GoosRub}.

\begin{maintheo*}
\ 
\begin{enumerate} 
\item The communication complexity of local search on the $n$-dimensional hypercube with $N=2^n$ vertices is $\Omega(\sqrt{N})$.
\item The communication complexity of local search on a constant-dimension grid with $N$ vertices is $\Omega(\sqrt{N})$. 
\item The communication complexity of local search on a specific family of constant degree graphs with $N$ vertices is $\Omega(\sqrt{N})$.
\item The communication complexity of local search on the odd graph with $N$ vertices is $\Omega(\sqrt[4]{N})$.
\end{enumerate} 
\end{maintheo*}
We note that all our bounds hold for \emph{randomized} communication complexity.  Interestingly, the first three bounds are optimal: first, since for these families of graphs an algorithm by \cite{Ald} finds a local optimum with $O(\sqrt N)$ queries in expectation, which clearly implies an analogous communication algorithm with the same efficiency. 



Our proof starts from considering the communication variant of a pebbling game \cite{GP14}. $D=(V,E)$ is a known directed acyclic graph. The input is a boolean assignment for the vertices $b:V \ra \{0,1\}$ such that every source is true ($b(v)=1$) and every sink is false ($b(v)=0$). The output is a false vertex whose all predecessors are true (i.e., $v\in V$ such that $b(v)=0$ and $b(u)=1$ for all $u\in V$, $(u,v)\in E$). \citep{GP14} consider the communication variant of the game which is obtained by distributing the information $b(v)\in \{0,1\}$ of every vertex by a constant size index-gadget $\{0,1\}^3\times [3] \ra \{0,1\}$. They show that for some constant-degree graph $D$ with $N$ vertices the communication complexity of the problem is $\Theta(\sqrt{N})$, which is optimal. 

Our proof is composed of three steps. The first step shows how to reduce the pebbling game to a variant of local search on a graph $G$ ($\ValLS$) where Alice holds 
the function $f$ and Bob holds a set of valid vertices. The goal is to find a local maximum in the subgraph that is composed of the valid vertices. 

The second step is the most technically challenging one. We first define a notion of embedding one graph to the other, and show that if a graph $G$ can be embedded into $H$ then the communication of $\ValLS(H)$ is at least that of $\ValLS(H)$. We then show that the graph $G$ obtained in the previous step can be embedded into each of the families considered in the theorem. This embedding is quite delicate and uses specifics properties of the graph $G$, since the number of vertices of $G$ and $H$ must be almost the same, in order to obtain an optimal bound of $\Omega(\sqrt N)$ for $\ValLS(H)$, where $N$ is the number of vertices of $H$.

Finally, in the third step we show that the communication complexity of $\ValLS$ on any graph is at least that of local search, thus establishing the theorem.

The constants that are obtained in our theorem are quite big (the dimension of the grid has to be at least $119$, and the degree of the constant degree graph is $36$). Thus, we also provide an alternative proof that obtains better constants, at the cost of a worse communication bound. Specifically, we show that there exists a specific family of $4$-degree graphs for which the communication complexity of local search is $\Omega(N^c)$ for some constant $c>0$. We also show a lower bound of the form $\Omega(N^c)$ for the three dimensional grid $N\times N\times 2$. The alternative proof uses the more recent and more generic ``simulation'' lemmas that ``lift'' lower bounds from the query complexity setting to the communication complexity setting \citep{GPW,GPW15,RM}, instead of the ``simulation'' lemma of \citep{GP14} that was developed for specific settings like the pebbling game.  The main technical difficulty that we overcome is that the ``combination gadgets'' used in these lemmas (specifically the index function) are very different from the simple sum that we desire.

We now describe two applications of our basic lower bound. In both applications we study communication variants of problems that are known to be PLS complete, have low non-deterministic complexity and, as we show, high communication complexity.

\subsection{Potential Games}

The communication requirements for reaching various types of equilibria in different types of games have received a significant amount of recent interest (\cite{BR,GoosRub}) as they essentially capture the convergence time of arbitrary dynamics in scenarios where each player only knows his own utilities (``uncoupled dynamics'' \citep{HMas,HMan}) and must ``learn'' information about the others. Of particular importance here is the class of potential games \citep{MS}.  

\vspace{0.1in}
\noindent
{\bf Definition:} An $n$-player game with strategy sets $A_1,...,A_n$ and utility functions $u_1,...,u_n$
is an \emph{exact potential game} if there exists a single potential function 
$\phi : A_1 \times \cdots \times A_n \rightarrow \R$ so that for every player $i$, every two strategies $a_i,a'_i \in A_i$ and every tuple of strategies $a_{-i} \in A_{-i}$ we have that $u_i(a_i,a_{-i})-u_i(a'_i,a_{-i})=\phi(a_i,a_{-i})-\phi(a'_i,a_{-i})$.

The game is an \emph{ordinal potential function}  if there exists a single potential function 
$\phi : A_1 \times \cdots \times A_n \rightarrow \R$ so that for every player $i$, every two strategies $a_i,a'_i \in A_i$ and every tuple of strategies $a_{-i} \in A_{-i}$ we have that $sign(u_i(a_i,a_{-i})-u_i(a'_i,a_{-i}))=sign(\phi(a_i,a_{-i})-\phi(a'_i,a_{-i}))$, i.e., the value of the potential function increases if and only if the player improves his utility.
\vspace{0.1in}

The class of exact potential games includes, in particular, all congestion games. A key property of potential games (exact or ordinal) is that every sequence of better responses converges to an equilibrium and therefore every potential game always has a pure Nash equilibrium.

\citep{HMan} study the communication complexity of pure Nash equilibrium in \emph{ordinal} potential games. They consider $n$-player games where each player has four actions and show (by a reduction from disjointness) that exponential communication is required to distinguish between the case where the game is an ordinal potential game (and thus has a Nash equilibrium) and the case where the game is not a potential game and does not admit any Nash equilibrium. This immediately implies that finding an equilibrium in games that are guaranteed to have one takes $exp(n)$ bits of communication.

Does finding an equilibrium become any easier for \emph{exact} potential games? In \citep{Noam-blog-2} it was shown that exponentially many \emph{queries} are needed to find an equilibrium, but maybe in the communication model the problem becomes much easier. The technical challenge is again that the non-deterministic communication complexity of the problem is low, i.e, verifying that a certain profile is a Nash equilibrium does not require much communication (each player only has to make sure that he plays his best response). Nevertheless, we provide a ray of hope and show that in contrast to ordinal potential games, there is a randomized protocol that uses only $\polylog(|A|)$ (when $|A|=|A_1|\cdot ... \cdot |A_n|$ is the game size) bits of communication and determines whether the game is an exact potential game or not.

We then show that although it is easy to recognize whether a game is an exact potential game or not, finding an equilibrium requires polynomial (in the size of the game) communication (and in particular exponential in the number of players). These results provide a negative answer to an open question posed in \citep{Noam-blog}.

\begin{theo}
For some constant $c>0$, the following problem requires at least $N^c$ communication (even randomized):  Alice gets an $N \times N$ matrix $u_A$ and Bob gets an
$N \times N$ matrix $u_B$, they are promised that the game defined by these matrices is an (exact) potential game and they must output a pure Nash equilibrium of the game.
\end{theo}
\begin{theo}
For some constant $c>0$, the following problem requires at least $2^{cn}$ communication (even randomized):  Alice gets the utility functions of the first $n$ players in a $2n$-player $2$-action game. Bob gets the utility functions of the last $n$ players. They are promised that the game defined by these matrices is an (exact) potential game and they must output a pure Nash equilibrium of the game.
\end{theo}

\noindent Our proofs are via reductions from local search on (certain) degree 4 graphs
in the two-player $N$-action case, and from local search on the hypercube in
the $2n$-player 2-action case.  While the relation between equilibria of potential games and local maxima is well known and very simple, the reduction is actually quite subtle.
First the neighbourhood structures do not naturally match (in the two-player case), but more crucially
the input to the players here is very limited: only very specifically related matrices $u_A$ and $u_B$ give an (exact) potential game, while the lower bounds for local search were for arbitrary inputs.

We also show that the search for a pure Nash equilibrium in exact potential games can be formulated as a \emph{total search problem}: Either find a pure Nash equilibrium (that is guaranteed to exist in exact potential games) or provide a succinct evidence that the game is not an exact potential game. Interestingly such a succinct evidence of violation of exact potential property is guaranteed to exist by \citep{MS}. As an immediate corollary from our results we deduce hardness of this total search problem.

\subsection{Local Optima in Combinatorial Auctions}

Our second application concerns attempts to weaken the global optimality constraints in market allocations. Consider a combinatorial auction of $m$ indivisible items among $n$ players, each with his own valuation function $v_i$ that gives a real value to every subset of the items.  The usual goal of optimizing social welfare aims to globally maximize $\sum_i v_i(S_i)$ over all allocations $(S_1,...,S_n)$ of the items.

A corresponding notion of equilibrium is the Walrasian equilibrium, which includes also a vector of  prices $p_1,...,p_m$ such that every player receives his globally-optimal set of items at these prices.  
While these notions provide very strong guarantees, they are usually ``too good to be true'': Walreasian equilibria only rarely exist and optimizing social welfare is 
usually infeasible, in essentially any sense of the word, and in particular in the sense of requiring exponential communication \citep{nisan2006communication}.

Several papers have tried to relax the notion of a Walrasian equilibrium or similarly view the allocation problem as a game and analyze the equilibria in this game. In particular, in the model of simultaneous second price auctions \citep{christodoulou2008bayesian} it is easy to see that when the valuations are submodular every allocation that is \emph{locally optimal} can be part of an equilibrium in the game, and the same goes for the endowed equilibrium of \citep{BDO18}. Recall that a locally optimal allocation in a combinatorial auction is an allocation of the items $(S_1,\ldots, S_n)$ such that transferring any single item $j\in S_i$ to some other player $i'$ does not improve the welfare.

Since local optima play a central role in various relaxed notions of equilibria, an obvious question is whether they are easy to find. In \citep{BDO18} it is shown that  for some succinctly represented submodular valuations it is PLS hard to compute a locally optimal allocation in combinatorial auction. Furthermore, in the query model it is shown that finding a locally optimal allocation is as hard as finding a local maximum in the odd graph. Combining the same reduction with our communication hardness of local search on the odd graph, we get that: 

\begin{theo} 
The communication complexity of finding a locally optimal allocation between two players with submodular valuations is $2^{\Omega(n)}$.
\end{theo}

%% file: graphs.tex
\section{Local Search over Graphs}

In this section we provide communication lower bounds on the communication complexity of local search over several families of graphs. 

\begin{theorem}\label{theo:opt}
The following bound holds for the randomized communication complexity of $\SumLS$:
\begin{enumerate}
\item\label{theo:opt-bounded} $CC(\SumLS(G))= \Omega(\sqrt{N})$, when $G$ is a specific constant-degree (36) graph with $N$ vertices. 
\item\label{theo:opt-hypercube} $CC(\SumLS(\Hyp_n)) = \Omega(\sqrt{N}) = \Omega(2^{n/2})$ where $N=2^n$ is the number of vertices.
\item\label{theo:opt-grid} $CC(\SumLS(H)= \Omega(\sqrt{N})$, when $H$ is a grid with a constant dimension (119) grid with $N$ vertices.
\item\label{theo:odd} $CC(\SumLS(\Odd_n)) = \Omega(2^{n/2})$.
\end{enumerate} 
\end{theorem}
We note that results \ref{theo:opt-bounded}, \ref{theo:opt-hypercube}, and \ref{theo:opt-grid} are optimal since \cite{Ald} provides a randomized algorithm that finds a local maximum in these graph using $O(\sqrt N)$ value queries. 
Result \ref{theo:odd}, on the other hand, is not necessarily optimal because the odd graph has $N\approx 4^{n}$ vertices, so in terms of the number of vertices our lower bound is $\Omega(\sqrt[4]{N})$.

Result \ref{theo:opt-grid} proves an optimal bound for a grid with a constant dimension, but this dimension  is quite large (119). We are able to show that finding a local optimum in the three-dimensional grid is hard, but our lower bound in this case is only $\Omega(N^c)$, for some constant $c>0$ (in contrast to an optimal bound of $\Omega(\sqrt N)$ for the $119$-dimensional grid). To prove this, we first show that finding a local maximum is hard even for degree $4$ graphs.
\begin{theorem}\label{theo:grid}
There exists a constant $c>0$ such that the randomized communication complexity of $\SumLS$ satisfies:
\begin{enumerate}
\item\label{theo:deg-bounded} $CC(\SumLS(G))\geq N^c$, when $G$ is a specific degree 4 graph with $N$ vertices. 
\item\label{theo:dim-grid} $CC(\SumLS(\Grid_{N\times N \times [2]})= N^c$.
\end{enumerate} 
\end{theorem}
The overall structure of the proofs of Theorems \ref{theo:opt} and \ref{theo:grid} is similar, but the proofs use different techniques.
The proof of Theorem \ref{theo:opt} appears in Section \ref{sec:main-pr} and the proof of Theorem \ref{theo:grid} appears in Section \ref{sec:grid-pr}. 

\section{Proof of Theorem \ref{theo:opt}}\label{sec:main-pr}

Our starting point is a communication variant of a pebbling game. In this problem, 
$D=(V,E)$ is a known directed acyclic graph. The input is a boolean assignment for the vertices $b:V \ra \{0,1\}$ such that every source is true ($b(v)=1$) and every sink is false ($b(v)=0$). The output is a false vertex whose all predecessors are true (i.e., $v\in V$ such that $b(v)=0$ and $b(u)=1$ for all $u\in V$, $(u,v)\in E$). Note that the problem is total.

The communication variant of the pebbling game $\textsc{Pebb(D)}$ is defined by distributing the information $b(v)\in \{0,1\}$ of every vertex by a 
constant size index-gadget $\{0,1\}^3\times [3] \ra \{0,1\}$. 

In \citep{GP14} it is shown that there exists a constant degree graph $D$ with $N$ vertices where both the randomized communication complexity of the problem is $\Theta(\sqrt{N})$. The proof is done in three steps.

\paragraph{Step 1} We introduce an intermediate communication problem $\ValLS(G)$ where Alice holds the potential function and Bob holds a subset of valid vertices (equivalently, Bob vetoes the vertices that are not in the set that he holds). The goal is to find a valid local maximum: a valid vertex whose \emph{valid} neighbours have (weakly) lower potential. Given a graph $D$ as above, we construct a constant degree graph $G$ with $O(N)$ vertices and reduce $\textsc{Peb(D)}$ to $\ValLS(G)$. This gives us an optimal communication lower bound for $\ValLS$ for a concrete graph $G$.

\paragraph{Step 2} We define a certain notion of ``embedding'' of one graph into the other. We show that if $G'$ can be embedded in $G$, then $CC(\ValLS(G'))\geq CC(\ValLS(G))$. We use this observation to prove optimal hardness bound of $\ValLS$ over the hypercube by embedding $G$ into an hypercube of dimension $\log(N)+c$ for a constant $c$. 

\paragraph{Step 3} For every graph $G$, we show that $CC(\ValLS(G))\approx CC(\SumLS(G))$.
 
\vspace{0.1in}\noindent We now provide a detailed description of each step.

\subsection{Starting Point: Pebbling Games}\label{sec:step0}
In Section \ref{sec:emb} we use the concrete structure of the constant degree graph $D$ for which the hardness of the pebbling game is proved. Hence, we start with providing an explicit description of the graph $D$.

The vertices of $D$ are given by $V=[M^3]\times [M]\times [M] \times [M]$. Each vertex $v=(k_1,k_2,k_3,k_4)$ has six successors:
$$\{v+(1,\pm 1,0,0),v+(1,0,\pm 1,0),v+(1,0,0,\pm 1)\}$$
where the $\pm 1$ addition in the last three coordinates is done modulo $M$. The addition in the first coordinate is the standard addition. Thus, each vertex has six predecessors:
$$\{v+(-1,\pm 1,0,0),v+(-1,0,\pm 1,0),v+(-1,0,0,\pm 1)\}$$
The sources of the graph are $\{(1,\cdot,\cdot,\cdot)\}$ and its sinks are $\{(M^3,\cdot,\cdot,\cdot)\}$. 

In \citep{GP14} an optimal bound on the communication complexity is proved the following optimal bound on the communication of the following variant of pebbling games. In the communication problem $\textsc{Pebb}(D)$, Alice's input is an assignment $b:V\times [3] \ra \{0,1\}$. Bob's input is an index for each vertex $I:V \ra [3]$. The input satisfies $b(v,I(v))=1$ for every source $v$, and $b(v,I(v))=0$ for every sink $v$. The output is a vertex $v\in S$ such that $b(v,I(v))=0$ and for every predecessor $u$ of $v$ holds $b(u,I(u))=1$.

\begin{theorem}[\citep{GP14}]\label{theo:gp}
$CC(\textsc{Pebb}(D))=\Omega(M^3)$. This bound also holds for randomized protocols.
\end{theorem}

\subsection{Step 1: From Pebbling to $\ValLS$}\label{sec:veto}

Given a graph $G$, the communication problem $\ValLS(G)$ is defined as follows. Alice's input is a function $f:V\ra [W]$. Bob's input is a non-empty subset $S\subset V$. The output is a vertex $v\in S$ such that $f(v)\geq f(w)$ for every $w\in S$ such that $\{v,w\}\in E$ (i.e., for every valid neighbour). We show that the communication complexity of $\ValLS(G)$ is at least that of $\textsc{Pebb}(D)$, for some $G$ that is related to $D$. Next we show how to obtain the graph $G$ from $D$.

We construct the graph $G$ in two stages. First, given a graph $D$ of the pebbling game, let $G'$ be an undirected version of $D$ which additionally has an edge from every source of $D$ to some sink of $D$. Let $G$ be the graph that is obtained from $G'$ by replacing each vertex in $G$ with three new vertices and duplicating the edges so that each new vertex is connected to all the copies of its neighbors in $G'$. We call the graph $G$ the \emph{replication graph} of $D$.  

\begin{proposition}\label{pro:peb-ls}
Let $D$ be a graph and $G$ be its replication graph. The communication complexity of $\ValLS(G)$ is at least that of $\textsc{Pebb}(D)$.
\end{proposition}
\begin{proof}
Let $V$ denote that set of vertices of $D$ and $V'=V\times [3]$ be the set of vertices of $G$. Let $t:V'\ra \mathbb R$ be a topological numbering of the vertices of $G$. I.e., $t((u,i))>t((v,j))$ if there exists a directed edge $(u,v)$ in $D$.

Alice's input in $\textsc{Pebb}(D)$ is an assignment $b:V\times [3] \ra \{0,1\}$. We use this to define the potential function $f$ that Alice holds in $\ValLS(G)$: for vertex $v\in V'$ let $f(v)=t(v)+6N\1_{b(v)=0}$. Bob's input in $\textsc{Pebb}(D)$ is the function $I:V \ra [3]$. Bob defines the set of valid vertices in $\ValLS(G)$ to be $S=\{(v,I(v)):v\in V\}$. Namely, among the three copies of $v$ only the one with correct index is valid. This choice of valid vertices has the desirable property that the subgraph of \emph{valid} vertices is precisely $G'$ and the assignment $b(v,i)$ over the vertices of $G'$ is precisely the decomposed assignment $b(v,I(v))$.

We argue that the local maxima of $f$ are precisely all false vertices whose all incoming neighbours are true. Those are indeed local maxima, because their ``predecessors" do not have the bonus of $6N$ and their ``successors" have lower topological number. The source cannot be a local maximum because it is a ``true'' vertex and it is connected to a sink that is a ``false'' vertex. A true vertex (other than source) is not local maximum because its predecessor has higher topological number. Similarly, a false vertex with false predecessor is not local maximum. This leaves us only with false vertices whose predecessors are true. 
\end{proof}

\subsection{Step 2: Embedding the Bounded Degree Graph}\label{sec:emb}
In this step we define a certain notion of embedding of one graph into another. We will see that if a graph $G$ can be embedded into $H$ then the communication complexity of local search on $H$ is essentially at least as large as the communication complexity of local search on $G$. We will then see how to embed the graph $G$ of the previous steps into the three dimensional grid, the hypercube, and the odd graph.

\begin{definition}\label{def:vied}
A \emph{vertex-isolated edge-disjoint (VIED) embedding} of a graph $G=(V_G,E_G)$ in a graph $H=(V_H,E_H)$ is a pair of mappings $\varphi : V_G \ra V_H$ and $\chi: E_G \ra P(H)$, where $P(H)$ is the set of \emph{simple paths} on $H$, such that:
\begin{itemize}
\item $\varphi$ is injective.
\item For every edge $\{v,w\}\in E_G$, the path $\chi(\{v,w\})$ connects $\varphi(v)$ to $\varphi(w)$.
\item The interior vertices of the paths $\chi(\{v,w\})$ and $\chi(\{v',w'\})$ are disjoint (edge disjointness).
\item For every $v\in V_G$ and every $\{w,w'\}\in E_G$ such that $v\neq w,w'$ holds $d(\varphi(v),\chi(\{w,w'\}))\geq 2$, where $d$ denotes the distance in $H$ of the vertex from the path (vertex isolation).
\end{itemize}
\end{definition}
That is, in a VIED embedding every edge of $G$ is replaced by a path in $H$ that connects the corresponding vertices such that these paths do not share a vertex. Moreover, for every $v\in V_G$, $\varphi(v)$ is isolated in the sense that no path passes through the neighbours of $\varphi(v)$.  

\begin{lemma}\label{lem:emb}
Let $G$ be a graph and suppose it can be VIED embedded into some other graph $H$. Then $CC(\ValLS(G))\leq CC(\ValLS(H))$.
\end{lemma}

\begin{proof}

Alice's potential is defined as follows. For vertices $w\in \varphi(V_G)$ we define  $f_H(\varphi(v))=f_G(v)$. Consider a vertex $w\in \chi(E_G)$ that belongs to an edge $\{u,v\}\in E_G $. Suppose that $w$ is the $k$'th element in the path $\chi(\{u,v\})$ and $l$ is the total length of this path. Define:
\begin{align}\label{eq:edge}
f_H(w)=\frac{k}{l} f_G(u) + \frac{l-k}{l}f_G(v)
\end{align}

In all other vertices Alice's potential will not play a role because these vertices will not be valid, thus we can simply set $f_H(w)\equiv 0$ for all other vertices.

We recall that Bob's input in $\textsc{VetoLS}(G)$ is $S_G\subset V_G$.
We denote by $E_G(S_G)\subset E_G$ the set of internal edges of $S_G$. 
Bob's subset of valid vertices in $H$ is defined by\footnote{By $\chi(E_G(S_G))$ we obviously mean the corresponding \emph{vertices} in these paths.} $S_H=\varphi(S_G) \cup \chi(E_G(S_G))$.

If $v\in V_G$ is a valid local maximum, then $\varphi(v)\in V_H$ is a valid local maximum because all its valid neighbours are valid edges in which $v$ participates (here we use the isolation property), and the value along these edges is a weighted average of $f_G(v)$ and $f_G(u)\leq f_G(v)$, where $u$ is a valid neighbour of $v$.

We argue that there are no additional valid local maxima in $H$. Indeed, if $v\in V_G$ is not a local maximum then $\varphi(v)\in V_H$ is not a local maximum because there is a valid edge where the potential increases. If $w\in \chi(E_G(S_G))$, by distinctness, $f_G(u)\neq f_G(v)$ therefore in one of the directions of the path $\chi(\{u,v\})$ the potential increases. All other vertices are invalid.
\end{proof}

\subsubsection{An Explicit Description of the Graph G}

In the embeddings we use the specifics of the DAG $D$ for which the hardness of pebbling games is proved in \cite{GP14}. We now explicitly describe the replication graph $G$ that is obtained from $D$ so that Proposition \ref{pro:peb-ls} can be applied.

Let $G'$ be the undirected version of the DAG $D$ for which the hardness of pebbling games is proved with additional edges that connect the sources and sinks of $D$ in a same way other vertices in $D$ are connected. Formally, the vertices of $G'$ are $V=[M^3]\times [M]\times [M] \times [M]$, and the edges are:
$$E=\{(u,v):u-v\in \{(\pm 1,\pm 1,0,0),(\pm 1,0,\pm 1,0),(\pm 1,0,0,\pm 1)\}\}$$

Let $G$ be the graph that is obtained from $G'$ by replacing each vertex in $G$ with three new vertices and duplicating the edges so that each new vertex is connected to all the copies of its neighbors in $G'$. Formally, the vertices of $G$ are $\{(v,i):v\in V,i\in [3]\}$ and the edges are $\{((u,i),(v,j)):(u,v)\in E, i,j\in [3]\}$. Note that $G$ is a graph with $3M^6$ vertices and (constant) degree $d=36$.

\subsubsection{Embedding into the Hypercube}

In this section we show how to embed the replication graph $G$ obtained in the previous step into the hypercube. Moreover, the embedding is such that the number of vertices in the hypercube increases only by a constant factor. This small blowup is crucial for obtaining an optimal $2^{n/2}$ bound.

\begin{lemma}\label{lem:hyp}
The graph $G$ (with $3M^6$ vertices) can be VIED-embedded into the $n$'th-dimensional hypercube $\Hyp_n$ for $n=6\lceil \log M \rceil + 111$. As a corollary, $CC(\ValLS(\Hyp_n))=\Omega(2^{n/2})$.
\end{lemma}

\begin{proof}
For clarity of exposition we assume that $M=2^c$ is a power of 2. 
We start with some notations and properties of the graph $G$. Recall that the vertices of $G$ are $V=[M^3]\times [M] \times [M] \times [M]\times [3]$.
For a vertex $v=(k_1,k_2,k_3,k_4,i)$, $k_1$ is called the \emph{layer of $v$}. Note that all edges connect $k$ layer vertices to $k+1$ layer vertices. $k_1+k_2+k_3+k_4 \mod 2$ is called the \emph{parity of $v$}. $i$ is called the \emph{replication index of $v$}.
We present an edge coloring of $G$ with $108$ colors in which no two adjacent edges are colored the same (a ``valid'' coloring).
We first color all edges from layer 1 to layer 2 with $54$ colors. Given a vertex $v=(k_2,k_3,k_4)$, edges are specified by a displacement $d\in \{\pm 1,0,0),(0,\pm 1,0),(0,0,\pm 1)\}$ that operates on $(k_2,k_3,k_4)$ and pair of replication indices $i,j\in [3]$ ($i$ is the replication index of the vertex at layer $1$ and $j$ is the replication index of the vertex at layer $2$). Note that we have $6\cdot 9=54$ such specifications. It is easy to verify that coloring these edges in $54$ different colors is a valid edge coloring.
We proceed by coloring all edges between layers $2$ and $3$ with \emph{different} $54$ colors using a similar coloring method. Similarly, all edges from layer $2k-1$ to layer $2k$ are colored as edges between layers $1$ and $2$ and all edges from layer $2k$ to layer $2k+1$ are colored as edges between layers $2$ and $3$.
This defines an edge coloring of $G$.

Now we present some notation. The vertices of the hypercube are partitioned into blocks as follows:

\begin{itemize}
\item For $i=1,...,5$ the \emph{$i$'th index} block consists of bits that represent the $i$'th index. The sizes of the blocks are $(3c,c,c,c,2)$ for $i=1,2,3,4,5$ correspondingly.
\item A \emph{parity bit} memorizes the parity of a vertex. 
\item The \emph{edge} block consists of $108$ bits.
\item The \emph{counter} block consists of $3$ bits that serves as a counter to keep track of the block on which we currently apply the changes along the embedding path (see below).
\end{itemize}

\paragraph{Embedding the vertices.} 
Let $(h_1,...,h_{M^3})$ be a Hamiltonian path of the $3c$-dimensional hypercube. Let $(h'_1,...,h'_{M})$ be a Hamiltonian path of the $c$-dimensional hypercube and $(h''_1,...,h''_{4})$ be a Hamiltonian path of the $2$-dimensional hypercube. To define $\phi(v)$, we embed a vertex $v=(k_1,k_2,k_3,k_4,i)$ into the vertex of the hypercube whose first block is the bits of $h_{k_1}$, the second block is $h'_{k_2}$, then $h'_{k_3}$, $h'_{k_4}$ and $h''_{k_5}$. We set the parity bit to be the parity of $v$, the edge block to $\0$, and the counter block to $\0$.


\paragraph{Embedding the edges.} 
Note that the coloring of $G$ in $108$ colors naturally induces an order on the edges. Every vertex has at most one $m$'th edge, and two adjacent vertices agree on the index of this edge.  
The $m$'th edge of $v$, from $v$ in layer $k_1$ to $u$ in layer $k_1+1$, is defined by the following sequence of bit flipping.
\begin{enumerate}
\item The $m$'th bit in the edge block is flipped to 1.
\item A single bit in the counter block is flipped to encode the integer 1.
\item A single bit in the first index block is flipped to encode the integer $k_1+1$.
\item A single bit in the counter block is flipped to encode the integer 2.
\item If the displacement of the edge is $(\pm 1,0,0)$, a single bit in the second index block is flipped to encode the integer $k_2 \pm 1$. If the displacement of the edge is $(0,\pm 1,0)$, a single bit in the third index block is flipped to encode the integer $k_3 \pm 1$. If the displacement of the edge is $(0,0,\pm 1)$, a single bit in the fourth index block is flipped to encode the integer $k_3 \pm 1$.
\item A single bit in the counter block is flipped to encode the integer 3.
\item The two bits of the fifth index block are flipped (one by one in a fixed order) to encode the integer $j$ (the replication index of $u$).
\item The counter block returns back to $\0$.
\item The $m$'th bit in the edge block is flipped back to $\0$.
\end{enumerate}

It is easy to see that this path ends up at $\phi(u)$ (note that the parity of $v$ and $u$ is the same, and indeed we did not flip the parity bit). We argue that the defined paths are disjoint. It is sufficient to prove that given a node on the path one can recover the previous node. Given the color of the edge and the counter, it is immediate to recover the previous node in all intermediate steps excluding steps (3) and (5). In steps (3) and (5) it is unclear whether we should flip the corresponding index block or the counter block. To determine this we use the parity bit: In step (3), if the parity bit is equal to the parity of the encoded vertices, then it means that we did not flip yet a bit, and to get the previous vertex we set the counter block to encode 0. If the parity bit differs from the parity of the encoded indices, then it means that we have flip a bit, and to get the previous vertex we should flip a bit in the index block. In step (5) we do the opposite. If the parity bit differs from the parity of the encoded indices, then we flip the counter. If the parity bit is equal to the parity of the encoded indices, then we flip the index block.

It is easy to check that the embedding is vertex isolated because of the parity bit.

\end{proof}

\subsubsection{Embedding into the Grid}

\begin{lemma}\label{lem:grid1}
The graph $G$ (with $3M^6$ vertices) can be VIED-embedded in a constant-dimension grid with $O(M^6)$ vertices. As a corollary, $CC(\ValLS(\Grid_d))=\Omega(\sqrt{N})$ for some constant-dimension grid with $N$ vertices.
\end{lemma}

\begin{proof}(sketch)
The embedding is very similar to the one we presented in Lemma \ref{lem:hyp} for embedding into the hypercube. In the proof of Lemma \ref{lem:hyp} we only used the fact that the hypercube has an Hamiltonian cycle. For the grid, we will take advantage of the observation that the two-dimensional grid has an Hamiltonian cycle.

Specifically, a vertex $v=(k_1,k_2,k_3,k_4,i)$ is embedded into the vertex of the grid whose first block is the bits of that correspond to a Hamiltonian cycle on $\Grid_{M^{1.5}\times M^{1.5}}$, blocks $2-4$ are specified using the Hamiltonian cycle on $\Grid_{M^{0.5}\times M^{0.5}}$, and block $5$ using the Hamiltonian cycle on $\Grid_{2}\times \Grid_2$. We set the parity bit to be the parity of $v$, the edge block to $\0$, and the counter block to $\0$.
Applying very similar arguments to the proof of Lemma \ref{lem:hyp} we establish the embedding of $G$ into the grid $[M^{1.5}]^2\times [M^{0.5}]^6 \times [2]^{111}$.
\end{proof}

\subsubsection{Embedding into the Odd Graph}

\begin{lemma}\label{lem:odd}
There exists a VIED embedding of $\Hyp_n$ in $\Odd_{n+2}$. As a corollary,
$CC(\ValLS(\Odd_n))=\Omega(2^{n/2})$.
\end{lemma}

\begin{proof}
We first embed $\Hyp_n$ in $\Hyp_{n+1}$ simply by $\phi_1(v)=(v,0)$ and $\chi_1(\{v,w\})=\{(v,0),(w,0)\}.$ Obviously this embedding is edge disjoint (but not vertex isolated).

We now embed $\Hyp_{n+1}$ in $\Odd_{n+2}$. We refer to each vertex of $\Hyp_{n+1}$ as a subset $S\subset [n+1]$. We denote $S+n+1=\{i+n+1:i\in S\}$. We denote $T^c=[n+1]\setminus T$ (this notation will be relevant for subsets of $[n+1]$ rather than subsets of $[2n+3]$ as the vertices of $\Odd_{n+2}$). 
The embedding is defined by 
\begin{align*}
\phi_2(S)=&S\cup (S^c + n+1). \\
\chi_2(S,S\cup \{i\})=&S\cup  (S^c + n+1) \ra (S^c\setminus \{i\}) \cup (S+n+1) \cup \{2n+3\} \ra \\
& S \cup \{i\} \cup ((S\cup \{i\})^c + n+1)
\end{align*}
It is easy to check that this indeed defines a valid path on $\Odd_{n+2}$. All the defined paths are disjoint because given a vertex on a path $T\cup (T'+n)\cup \{2n+3\}$ we can identify the edge: $S=T'$ and $i$ is the unique element that is missing from both sets $T$ and $T'$.

Now we define the embedding of $\Hyp_n$ in $\Odd_{n+2}$ to be the decomposition of these two embeddings; I.e., $\phi(v)=\phi_2(\phi_1(v))$ and $\chi(e)=\chi_2(\chi_1(e))$. The embedding $(\phi,\chi)$ is edge disjoint because both embeddings $(\phi_1,\chi_1)$ and $(\phi_2,\chi_2)$ are edge disjoint. Now we prove that $(\phi,\chi)$ is vertex isolated. A vertex $\phi_2(\phi_1(v))=S\cup (S^c + 2n)$ has $n+2$ neighbours in $\Odd_{n+2}$. Among these neighbours, $n+1$ participate in an embedding of the outgoing edges of $S\in \Hyp_{n+1}$. So there is a single neighbour, $S^c \cup (S+n+1)$, who is suspected to belong to an embedding of an independent edge. Note that $S^c \cup (S+n+1)=\phi_2(S^c)$ and $S^c\in \Hyp_{n+1}$ \emph{does not} belong to the embedding of $\Hyp_n$ in $\Hyp_{n+1}$: indeed, for every vertex $v\in \Hyp_n$ the complementary vertex $\overline{(v,0)}=(\overline{v},1)\in \Hyp_{n+1}$ does not belong to the embedding of $\Hyp_n$ in $\Hyp_{n+1}$ (neither to $\phi_1(V_{\Hyp_n})$ nor to $\chi_1(E_{\Hyp_n})$). 
\end{proof}

\subsection{Step 3: From $\ValLS$ to $\SumLS$}\label{sec:veto-sum}
First, recall that the potential function gets values in $[W]$. We reduce the problem $\ValLS(G)$ to $\SumLS(G)$.
Alice's potential remains unchanged (i.e., $f_A(v):=f_G(v)$). Bob fixes some valid vertex $v^*\in S$ and sets his potential as follows: $f_B(v):=0$ if $v\in S$, otherwise he sets $f_B(v)=-d(v,v^*)\cdot (W+1)$, where $d$ is the distance in $G$. Indeed every valid local maximum $v$ is a local maximum of the sum because all the valid neighbours have lower sum of potentials $f_A(v)+f_B(v)=f_G(v)\geq f_G(w)=f_A(v)+f_B(v)$ and all invalid neighbours have negative sum of potentials $f_A(w)+f_B(w)\leq W - (W+1)<0$. It is easy to check that every valid vertex that is not a local maximum is not a local maximum of the sum. Finally, every invalid vertex $v$ is not a local maximum of the sum because the neighbour $w$ in the direction of the shortest path to $v^*$ has higher sum of potentials: 
\begin{align*}
f_A(v)+f_B(v) &\leq W - d(v,v^*)(W+1) < -(d(v,v^*)-1)(W+1) \\ &= -(d(w,v^*)-1)(W+1)\leq f_A(w)+f_B(w).
\end{align*}
We apply this reduction on the graphs considered in Lemmas \ref{lem:grid1}, \ref{lem:hyp} and \ref{lem:odd} to deduce the theorem.

\section{Proof of Theorem \ref{theo:grid}}\label{sec:grid-pr}

The overall structure of the proof is similar to that of Theorem \ref{theo:opt}. 
\paragraph{Step 0} We start with a local-search-related communicationally-hard problem over some graph $H$.
\paragraph{Step 1} We use the intermediate problem $\ValLS(G)$, where $G$ is constructed from $H$.
\paragraph{Step 2} We embed $G$ in the three-dimensional grid.
\paragraph{Step 3} We reduce $\ValLS(\Grid)$ to $\SumLS(\Grid)$.

\vspace{0.1in} \noindent However, in order to be able to embed $G$ in the three-dimensional grid, the degree of $G$ should be very low; at most 6.
The pebbling game result of \citep{GP14} does not serve our purposes because the degree of the graph $G$ is 36.
Hence, our starting point is some different local-search-related communicationally hard problem over some degree 3 graph $H$. In Step 1, we carefully modify $H$ to $G$ by increasing the degree only by 1; i.e., $G$ is degree 4 graph. Now, in Step 2 we are able to embed $G$ in the three-dimensional grid. Step 3 is identical to that in the proof of Theorem \ref{theo:opt}.

\subsection{Step 0: The Query Complexity of Local Search and its Simulated Variant}

In the problem $\QLS(H)$ there is a graph $H$ and a function $h$ that gives a value $h(v)$ for every vertex. The function $h$ can only be accessed via queries $h(v)$. Furthermore, for each two vertices $v,u$ are distinct: $h(v)\neq h(u)$. The goal is to find a local maximum of $h$ while minimizing the number of queries.

Santha and Szegedy \citep{SS} introduced a general connection between the query complexity of the local search problem and the expansion of a graph. Since random $3$-regular graphs are expanders with high probability, we have that there exists a degree 3 graph $H$ with $N$ vertices for which finding a local maximum requires $\poly(N)$ queries. However, their construction does not assume that $h(v)\neq h(u)$ for every two vertices $v$ and $u$. This is easy to fix: let $h'(v)=2N\cdot h(v) +v$ (where $v\in [N]$ denotes the index of $v$). Observe that each local maximum of $h'$ is also a local maximum of $h$ and that the query $h'(v)$ can be computed by one query $h(v)$, so the number of queries required to find a local maximum of $h'$ is at least the number of queries required to find a local maximum of $h$. We therefore have:

\begin{lemma}[essentially \citep{SS}]
There exists a degree 3 graph $H$ with $N$ vertices and a function $h'$ such that every vertex has a distinct value for which finding a local maximum requires $\poly(N)$ queries.
\end{lemma}   

The simulation theorems provides us a recipe to produce problems with high communication complexity, given a problem with high query complexity. In particular, \citep{GPW,AGJKM} suggest the \emph{index-gadget} recipe, which starting from $\QLS(H)$ is translated to the following communication problem $\SimLS(H)$: for each vertex $v\in H$, Alice holds an array of valuations $(f(v,i))_{i\in [M]}$ where $f(v,I(v))=h(v)$ and\footnote{E.g., $M=N^{256}$ in \citep{GPW}.} $M=\poly(N)$. Bob holds the correct index $I(v)\in [M]$. Their goal is to compute a local maximum of the function $f(v,I(v))$. Direct application of the simulation theorems to our setting gives that:
$$CC(\SimLS(H))=\Theta(\log N)QC(\QLS(H))=\poly(N)$$

\subsection{Step 1: The Communication Complexity of $\ValLS$}

In this step we prove the communication hardness of $\ValLS$ on a certain bounded degree graph. We recall the definition of $\ValLS(G)$. Alice's input is a function $f_G:V\ra [W]$. Bob's input is a non-empty subset $S\subset V$. The output is a vertex $v\in S$ such that $f_G(v)\geq f_G(w)$ for every $w\in S$ such that $\{v,w\}\in E$ (i.e., for every valid neighbour).

Unlike the communication pebbling game problem that uses index gadgets of size 3, the simulated $\QLS$ problem uses gadgets of size $M=\poly N$ ($N$ is the number of vertices of $G$). The idea in the proof of Theorem \ref{theo:opt} is to replicate each vertex according to the  gadget size, and connect every vertex with all its replicated neighbours. This idea is impractical here, because the degree of the resulting graph will be huge. Instead, we replace each replicated vertex with degree $3M$ by a carefully chosen binary tree structure in order to reduce the degree.

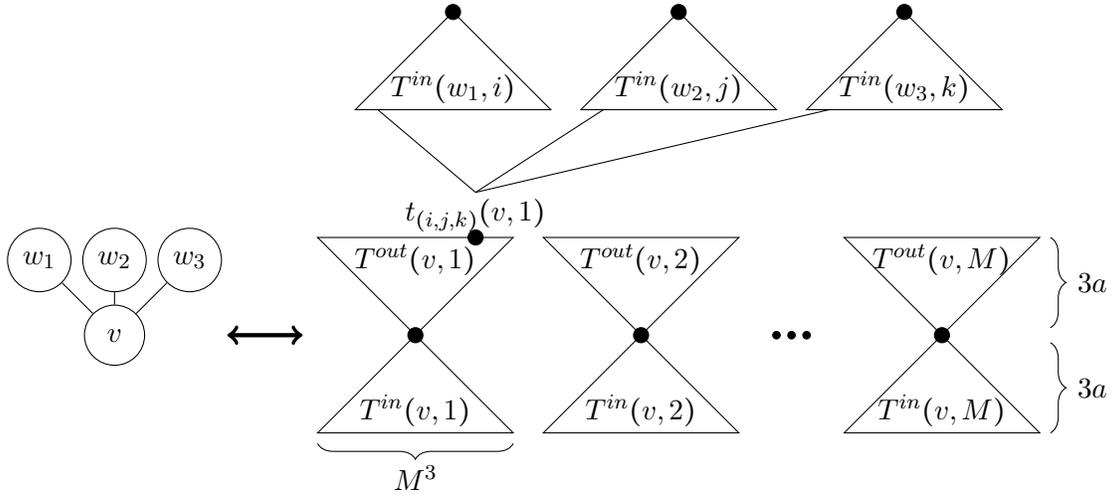
\begin{figure}[h]
\caption{The graph $G$. The replacement of a vertex by $M$ pairs of binary trees, and the neighbours of the leaves of $T^{out}$.}\label{fig:g}
\centering
\vspace*{3mm}
\begin{tikzpicture}
\node[circle, draw, minimum size=0.8cm] (v) at (-1,0) {$v$};
\node[circle, draw] (w1) at (-2,1) {$w_1$};
\node[circle, draw] (w2) at (-1,1) {$w_2$};
\node[circle, draw] (w3) at (0,1) {$w_3$};

\draw (v) -- (w1);
\draw (v) -- (w2);
\draw (v) -- (w3);

\draw[<->,ultra thick] (0.5,0) -- (1.5,0);

\filldraw[black] (3,0) circle (0.1);
\draw (3,0) -- (1.7,1.3) -- (4.3,1.3) -- (3,0);
\draw (3,0) -- (1.7,-1.3) -- (4.3,-1.3) -- (3,0);

\node at (3,1) {$T^{out}(v,1)$};
\node at (3,-1) {$T^{in}(v,1)$};

\filldraw[black] (6,0) circle (0.1);
\draw (6,0) -- (4.7,1.3) -- (7.3,1.3) -- (6,0);
\draw (6,0) -- (4.7,-1.3) -- (7.3,-1.3) -- (6,0);

\node at (6,1) {$T^{out}(v,2)$};
\node at (6,-1) {$T^{in}(v,2)$};

\filldraw[black] (7.8,0) circle (0.05);
\filldraw[black] (8,0) circle (0.05);
\filldraw[black] (8.2,0) circle (0.05);

\filldraw[black] (10,0) circle (0.1);
\draw (10,0) -- (8.7,1.3) -- (11.3,1.3) -- (10,0);
\draw (10,0) -- (8.7,-1.3) -- (11.3,-1.3) -- (10,0);

\node at (10,1) {$T^{out}(v,M)$};
\node at (10,-1) {$T^{in}(v,M)$};

\draw[decorate,decoration={brace,amplitude=6pt,mirror,raise=4pt},yshift=0pt]
(11.3,0.1) -- (11.3,1.3) node [black,midway,xshift=0.7cm] {
$3a$};

\draw[decorate,decoration={brace,amplitude=6pt,mirror,raise=4pt},yshift=0pt]
(11.3,-1.3) -- (11.3,-0.1) node [black,midway,xshift=0.7cm] {
$3a$};

\draw[decorate,decoration={brace,amplitude=6pt,mirror,raise=4pt},yshift=0pt]
(1.7,-1.3) -- (4.3,-1.3) node [black,midway,yshift=-0.6cm] {
$M^3$};

\filldraw[black] (3.8,1.3) circle (0.1);

\node at (3.8,1.6) {$t_{(i,j,k)}(v,1)$};

\filldraw[black] (3.5,4.3) circle (0.1);
\draw (3.5,4.3) -- (2.2,3) -- (4.8,3) -- (3.5,4.3);

\filldraw[black] (6.5,4.3) circle (0.1);
\draw (6.5,4.3) -- (5.2,3) -- (7.8,3) -- (6.5,4.3);

\filldraw[black] (9.5,4.3) circle (0.1);
\draw (9.5,4.3) -- (8.2,3) -- (10.8,3) -- (9.5,4.3);

\node at (3.5,3.3) {$T^{in}(w_1,i)$};
\node at (6.5,3.3) {$T^{in}(w_2,j)$};
\node at (9.5,3.3) {$T^{in}(w_3,k)$};

\draw (3.8,1.9) -- (2.5,3);
\draw (3.8,1.9) -- (5.5,3);
\draw (3.8,1.9) -- (8.5,3);

\end{tikzpicture}
\end{figure}

\paragraph{The Graph $G$.} Without loss of generality we assume that $M=2^a$ is a power of 2.
We obtain our graph $G$ by replacing every vertex $v\in H$ by a tuple of $M$ graphs $(T^{out}(v,i)\cup T^{in}(v,i))_{i\in M}$, 
where $T^{out}(v,i)\cup T^{in}(v,i)$ denotes two binary trees with an overlapping root, both of depth $\log (M^3)=3a$ (see Figure \ref{fig:g}). 
Roughly speaking, the role of $T^{out}(v,i)$ is to decode the correct indices of the three neighbours, 
and in parallel to split the outgoing edges from $v_i$. 
The role of $T^{in}(v,i)$ is simply to gather the incoming edges into $v_i$.

More formally, the vertices of $T^{out}(v,i)$ at depth $d$ are denoted by $(t_s(v,i))_{s\in \{0,1\}^d}$.
The vertices of $T^{in}(v,i)$ at depth $d$ are denoted by $(t'_s(v,i))_{s\in \{0,1\}^d}$. The vertices at depth $3a$ will be called \emph{leaves}\footnote{Note that they are leaves only with respect to the tree. In the graph $G$ they will not be leaves.}.
As was mentioned above, the vertex at depth 0 of these two trees coincides (i.e., $t_\emptyset (v,i) = t'_\emptyset (v,i)$). 
Now we describe how the leaves of $T^{out}(v,i)$ connect to the leaves of $T^{in}(w,j)$ for $w\neq v$.
For a leaf $t_s(v,i)\in T^{out}(v,i)$ we denote $s=(j_1,j_2,j_3)$ where $j_1,j_2,j_3\in [M]$ are the indices of the three neighbors of $v$, $w_1,w_2,w_3$. The leaf $t_s(v,i)\in G$ has a single edge to the tree $T^{in}(w_1,j_1)$, a single edge to the tree $T^{in}(w_2,j_2)$, and a single edge to the tree $T^{in}(w_3,j_3)$ (see Figure \ref{fig:g}).
In principle, we should specify which leaf exactly in $T^{in}(w_1,j_1)$ is connected to $t_s(v,i)$. However, 
since it will not play any role in our arguments, we just implement a counting argument to ensure that the 
number of neighbours from other trees of every leaf $t'_{s'}(w,j)$ is at most 3. If $w$ has a neighbour $v$, then for every $i\in [M]$
exactly $M^2$ vertices $t_s(v,i)$ will encode the index $j$. So from $T^{out}(v,i)$ we have $M\cdot M^2=M^3$ incoming edges. 
Summing over the 3 neighbours we get $3M^3$ incoming edges. If we distribute them equally among the $M^3$ vertices, we get 3 neighbours for each.

\paragraph{Alice's Potential. }
Alice's potential function is defined by $f_G(t'_s(v,i))=7af(v,i)+3a-|s|$ and $f_G(t_s(v,i))=7af(v,i)+3a+|s|$. Namely the potential in the tree $T'_s(v,i)$ starts at a value of $7af(v,i)$ in the leaves of $T^{in}(v,i)$. It increases by $1$ after every edge
until it gets to the root. At the root we move to the tree $T^{out}(v,i)$ where it proceeds to increase by $1$ until it gets to the leaves of $T^{out}(v,i)$ where the value of the potential is $7af(v,i)+6a$.

\paragraph{Bob's valid Vertices. }Now we define the subset of valid vertices $S$ held by Bob. 
Let $bin(i)\in \{0,1\}^a$ denote the binary representation of an index $i\in [M]$.
We denote by $nbin(v)=(bin(I(w_i)))_{i=1,2,3}$ the binary 
representation of the triple of $v$'s neighbours. For a binary string $b$ we denote by $b_{[k]}$ its first $k$ elements.
A vertex $t_s(v,i)\in S$ iff $i=I(v)$ and $s=nbin(v)_{[|s|]}$ (recall that $I(v)$ is Bob's input in $\SimLS$). Informally speaking the valid vertices are those where the tree $T^{out}(v,i)$ (or $T^{in}(v,i)$) has the correct index, and if the vertex is in $T^{out}(v,i)$ we require, in addition, that the prefix of the encoding of the neighbours' indices will be correct.

\paragraph{Local Maxima in $G$.} Since the potential of Alice increases starting from the leaves of $T^{in}(v,i)$ and ending at the leaves of $T^{out}(v,i)$, and in addition for every valid vertex there exists a valid neighbour with higher (lower) depth in $T^{out}(v,i)$ (in $T^{in}(v,i)$) the valid local maxima appear only on the leaves of $T^{out}(v,i)$. Every valid leaf of $T^{out}(v,i)$ has a potential of $7af(v,I(v))+6a$ (i.e., the correct potential) and is connected to leaves of $T^{in}(w_j,I(w_j))$ for $j=1,2,3$ with a potential of $7af(w_j,I(w_j))$ (i.e., the correct potential of the neighbours). Note that the potential values are integers. Therefore, $7af(v,I(v))+6a \geq 7af(w,I(w))$ if and only if $f(v,I(v))\geq f(w,I(w))$. Hence, there is a one-to-one correspondence between valid local maxima of $f_G$ with respects to the set if valid vertices $S$ and local maxima of $h$ over $H$.

This completes the proof item \ref{theo:deg-bounded} of the Theorem.

\subsection{Step 2: Embedding the Degree 4 Graph Into the Grid}
We VIED embed (see Definition \ref{def:vied}) the degree 4 graph $G$ obtained in the previous step into the grid. We use Lemma \ref{lem:emb} to deduce hardness of $\ValLS$ over the grid.

\begin{lemma}\label{lem:grid}
Every degree 4 graph $G$ with $N$ vertices can be VIED-embedded in $\Grid_{4N\times (2N+2) \times 2}$. As a corollary, $CC(\ValLS(\Grid_{N\times N \times 2}))=\poly(N)$.
\end{lemma}

\begin{proof}
We embed the graph $G$ in the grid whose vertices are $\{3,4,...,4N+2\}\times \{-1,0,...,2N\} \times \{0,1\}$. We denote the vertices of $G$ by $\{v_i\}_{i\in [N]}$ and we embed $\phi(v_i)=(4i,0,0)$. We use (for instance) the structure of Figure \ref{fig:paths} to place the four outgoing edges of $(4i,0,0)$ at the points $(4i-1,1,0),(4i,1,0),(4i+1,1,0)$ and $(4i+2,1,0)$.

\begin{figure}[h]
\caption{The outgoing edges of the embedded vertices.}\label{fig:paths}
\centering
\begin{tikzpicture}[scale=0.7]
\draw[step=1,gray,thin] (0,0) grid (16,3);
\draw[line width=3] (1,2) -- (1,0);
\draw[line width=3] (2,1) -- (0,1);
\draw[line width=3] (0,1) -- (0,2);
\draw[line width=3] (2,1) -- (2,2);
\draw[line width=3] (1,0) -- (3,0);
\draw[line width=3] (3,2) -- (3,0);

\draw[line width=3] (5,2) -- (5,0);
\draw[line width=3] (6,1) -- (4,1);
\draw[line width=3] (4,1) -- (4,2);
\draw[line width=3] (6,1) -- (6,2);
\draw[line width=3] (5,0) -- (7,0);
\draw[line width=3] (7,2) -- (7,0);

\draw[line width=3] (14,2) -- (14,0);
\draw[line width=3] (15,1) -- (13,1);
\draw[line width=3] (13,1) -- (13,2);
\draw[line width=3] (15,1) -- (15,2);
\draw[line width=3] (14,0) -- (16,0);
\draw[line width=3] (16,2) -- (16,0);

\filldraw (1,1) circle (0.3);
\filldraw (5,1) circle (0.3);
\filldraw (14,1) circle (0.3);

\node at (0,-0.5) {3};
\node at (1,-0.5) {4};
\node at (5,-0.5) {8};
\node at (14,-0.5) {$4N$};

\node at (-0.5,0) {-1};
\node at (-0.5,1) {0};
\node at (-0.5,2) {1};

\end{tikzpicture}
\end{figure}

We denote by $\{e_i\}_{i\in [m]}$ the edges in the graph $G$. Note that $m\leq 4N/2=2N$ because the graph degree is 4. The embedding of the edges is by an increasing order $e_1,...,e_m$. For an edge $e_i=(v_j,v_k)$ let $r_j\in \{-1,0,1,2\}$ be the minimal index such that the vertex $(4j+r_j,1,0)$ is not yet used by previous edges $\{e_{i'}\}_{i'<i}$. Similarly we define $r_k$. The edge $e_i=(v_j,v_k)$ is embedded to the path: 
$$(4j+r_j,1,0)\leftrightsquigarrow (4j+r_j,i,0) \leftrightarrow (4j+r_j,i,1) \leftrightsquigarrow (4k+r_k,i,1) \leftrightarrow (4k+r_k,i,0) \leftrightsquigarrow (4k+r_k,1,0)$$
where $(x,y,z) \leftrightsquigarrow (x,y',z)$ denotes a straight line that consistently changes  the second coordinate (similarly for $(x,y,z) \leftrightsquigarrow (x',y,z)$).

The embedding is VIED because all horizontal lines appear at $(\cdot,\cdot, 1)$ while all vertical lines appear at $(\cdot,\cdot, 0)$. The embedding is vertex isolated by the construction of Figure \ref{fig:paths}.
\end{proof}

Finally Step 3 is identical to Section \ref{sec:veto-sum}. We use the reduction from $\ValLS$ to $\SumLS$ to deduce the Theorem.

%% file: potential2n.tex
\section{The Communication Complexity of Exact Potential Games}

Recall that a game is an exact potential game if  there exists a potential function $\phi:A^n\rightarrow \mathbb R$, such that $\phi(a_i,a_{-i})-\phi(a'_i,a_{-i})=u_i(a_i,a_{-i})-u_i(a'_i,a_{-i})$ for every player $i$, every pair of actions $a_i,a'_i\in A_i$, and every profile of the opponents $a_{-i}\in A_{-i}$.
In this section we study the communication complexity of exact potential games. We assume that each of the players knows only his own utility function and the goal is to compute a pure Nash equilibrium in the game. In game theoretic settings this form of information distribution is called \emph{uncoupledness} \citep{HMas,HMan}. It is known that the communication complexity of computing an equilibrium captures (up to a logarithmic factor) the rate of convergence of uncoupled dynamics to equilibrium \citep{CS,HMan}.

As a preliminary result, we demonstrate that \emph{determining} whether a game is an exact potential games (under the uncoupled distribution of information) requires low communication. This result is in contrast to ordinal potential games (see Appendix \ref{ap:ident-ord}).

%
%

\begin{proposition}\label{pro:epd}
Consider a game with $n$ players and $N$ actions. There exists a randomized communication protocol that determines whether the game is an exact potential game or not that uses only $\poly(\log(N),n)$ bits of communication.
\end{proposition}

The proof is quite simple, and we demonstrate it here for $2$-player games. Monderer and Shapley \citep{MS} show that a two-player game $(A,B,u_A,u_B)$ is an exact potential game if and only if for every four actions $a,a'\in A$, and $b,b'\in B$ we have
\begin{align}\label{eq:cycle}
\begin{split}
&(u_A(a',b)-u_A(a,b))+(u_B(a',b')-u_B(a',b)) \\
&+(u_A(a,b')-u_A(a',b'))+(u_B(a,b)-u_B(a,b'))=0
\end{split}
\end{align}
Namely, the sum of gains/losses from unilateral divinations over every cycle of size four should sum up to zero. 
Now each player checks, for every possible four-action cycle, whether the sum of changes in his utility equals the negative of the change in utility of the other player for the same cycle. Verifying this simultaneously for all cycles can be done by applying any efficient protocol for the equality problem (we recall that we focus on \emph{randomized} communication protocols). For a general number of players, a similar characterization exists and we have to use protocols based on the ``equal sum'' problem as demonstrated below. 
\begin{proof}[Proof of Proposition \ref{pro:epd}]
By \citep{MS}, an $n$-player game $(A,u)$ is an exact potential game if and only if for every pair of permutations $\opi,\upi$ over $[n]$ and for every pair of action profiles $a,b\in A$ we have
\begin{align}\label{eq:n-cyc}
\begin{split}
&\sum_{k=1}^n u_{\opi(k)}(b,a,\opi([k]))-u_{\opi(k)}(b,a,\opi([k-1]))+ \\ 
&\sum_{k=1}^n u_{\upi(k)}(a,b,\upi([k]))-u_{\opi(k)}(a,b,\upi([k-1]))=0
\end{split}
\end{align}
Simply speaking, for every sequence of unilateral deviations that starts at $a$ goes back and forth to $b$, where each player changes his strategy from $a_i$ to $b_i$ once and from $b_i$ to $a_i$ once, the sum in the gains/losses of all players from the unilateral divinations should sum up to 0.

The players should check whether Equation \eqref{eq:n-cyc} holds for all possible pairs of profiles $a,b\in [N]^n$ and pairs of permutations $\opi,\upi$ over $[n]$. The number of these equations is $c=m^{2n} (n!)^2$. Each player can generate from his private input a vector in $\{-2W,...,0,...,2W\}^c$ which captures the sum of changes in his utility for each one of the tuples $(a,b,\opi,\upi)$. So the problem can be reduced to the following: 
Each player $i$ holds a vector $v_i\in \{-2W,...,0,...,2W\}^c$ and the goal of the players is to determine whether $\sum_{i\in [n]} v_i = \0_c$. This variant of the equality problem has a $\poly(\log W,\log c)=\poly(n,\log N)$ randomized communication protocol \citep{nisan1993communication,viola2015communication}. 
\end{proof}

In contrast, identifying whether a game is an \emph{ordinal} potential game is hard, even for randomized communication protocols. Identification of the ordinal potential property has a reduction to the disjointness problem. We relegate these reductions (for two-player and for $n$-player games) to Appendix \ref{ap:ident-ord}. 
The contrast between the hardness of identifying whether a game is an ordinal potential game and the easiness of identifying whether a game is an exact potential game might give some hope that computing an equilibrium in exact potential games is much easier than in ordinal potential games. Unfortunately, our main results for this section show that finding a Nash equilibrium remains hard even for exact potential games. 

\begin{theorem}\label{theo:2pot}
Consider the two-party promise communication problem where Alice holds the utility $u_A:[N]\times [N] \ra \R$, and Bob holds the utility $u_B$ of an exact potential game. The goal is to output a pure Nash equilibrium of the game. The problem requires $\poly(N)$ bits of communication, even for randomized protocols.
\end{theorem}
We can also show hardness for the $2n$-player $2$-action case.

\begin{theorem}\label{theo:n-pot}
Consider the two-party promise communication problem where Alice holds the utilities of $(u_i)_{i\in [n]}$ and Bob holds the utilities $(u_i)_{i\in [2n]\setminus [n]}$ of an exact potential game, and they should output a pure Nash equilibrium of the game.
The problem requires $2^{\Omega(n)}$ communication, even for randomized protocols.
\end{theorem}
This problem is obviously requires at least as much communication as the $2n$-party communication problem where each player holds his own utility function. 

In both theorems, we reduce from the problem of finding a local maximum (on a bounded degree graph in the two player case and on the hypercube in the $n$ player case) and show that the set of pure Nash equilibria corresponds exactly to the set of local maxima. The proofs of the Theorems appear in Sections \ref{sec:pr-2} and \ref{sec:pr-n}.

\subsection{Total variants of Pure Nash Equilibrium Search}\label{sec:tot}

In Theorems \ref{theo:2pot} and \ref{theo:n-pot} we have demonstrated communicational hardness of two \emph{promise} problems. Such hardness results are not rare in the literature. For instance, finding a pure Nash equilibrium in a game when it is promised that such an equilibrium exists.

To appreciate the novelty of our results we focus on a \emph{total} variant of equilibrium search problem $\textsc{TotExPot}$: either find a Nash equilibrium or provide a succinct evidence that the game is not an exact potential game. By \citep{MS} such a succinct evidence, in the form of a violating cycle (see Equations \eqref{eq:cycle},\eqref{eq:n-cyc}), necessarily exists. 
More formally, in the problem $\textsc{TotExPot}(2,N)$ Alice holds the utility $u_A$, Bob holds a utility $u_B$ of an $N\times N$ game, and the output is either a pure Nash equilibrium or a cycle of actions of size 4 that violates Equation \eqref{eq:cycle}. Similarly in the problem $\textsc{TotExPot}(2n,2)$ Alice holds the utilities $(u_i)_{i\in n}$, Bob holds the utilities $(u_i)_{i\in [2n]\setminus [n]}$ of an $2n$-player 2-action game, and the output is either a pure Nash equilibrium or a cycle of actions of size $4n$ that violates Equation \eqref{eq:n-cyc}.

In Proposition \ref{pro:epd} we showed that low communication is needed to determine whether a game is an exact potential game or not (accompanied with an evidence in case it is not). From these observation along with Theorem \ref{theo:2pot} we deduce that

\begin{corollary}
The total search problem $\textsc{TotExPot}(2,N)$ requires $\poly(N)$ communication.  
\end{corollary}

Similarly for the $2n$-player 2-action case we have

\begin{corollary}
The total search problem $\textsc{TotExPot}(2n,2)$ requires $2^{\Omega(n)}$ communication.  
\end{corollary}

Note that the non-deterministic complexity of $\textsc{TotExPot}(2,N)$ is $\log(N)$. Indeed a Nash equilibrium can be described by single action profile ($\Theta(\log N)$ bits), and a violating cycle can be described by $4$ action profiles. Each player can verify his best-reply condition and communicate a single bit to the opponent. Also verification of violating cycle can be done by communicating 4 valuations of utility.
Similarly, we can show that the non-deterministic complexity of $\textsc{TotExPot}(2n,n)$ is $\poly(n)$. Thus again, our results demonstrate an exponential separation between the non-deterministic and the randomized communication complexity of a total search problem.




\section{Proof of Theorem \ref{theo:2pot}}\label{sec:pr-2}
We reduce the problem of finding a local maximum on a graph $G$ with degree $4$ to finding a Nash equilibrium in an exact potential game with two players and $N$ actions. We then apply Theorem \ref{theo:grid}(\ref{theo:deg-bounded}) to get our communication bound. 

We construct the following exact potential game. For a vertex $v\in V$ we denote by $n_i(v)$ the $i$'th neighbour of $v$ for $i=1,2,3,4$. The strategy set of both players is $A=B=V\times [W]^5$ (recall that the potentials in $\SumLS(G)$ get values in $[W]$ and that $W=\poly(N)$). The interpretation of a strategy $(v,x)\in A$ where $\overrightarrow x=(x_0,x_1,...,x_4)\in [W]^5$ is (Alice's reported) potential for $v$ and its four neighbours. This report induces a valuation for all vertices $w\in V$ by 
\begin{align*}\label{eq:rep-val}
val^{(v,\overrightarrow x)}(w)=
\begin{cases}
x_0 & \text{ if } w=v; \\
x_i & \text{ if } w=n_i(v); \\
0   & \text{ otherwise.}
\end{cases}
\end{align*}
A strategy $(v,\overrightarrow x)$ is \emph{truthful} if and only if $x_0=f_A(v)$ and $x_i=f_A(n_i(v))$ for all neighbours of $v$ (in short, $\overrightarrow x=n(v)$). Similarly Bob's strategy $(w,\overrightarrow y)$ induces a valuation $val^{(w,\overrightarrow y)}(v)$ on all vertices $v\in V$, and a truthful report is similarly defined.
  
The utilities of Alice and Bob are given by (recall that $d(v,w)$ is the distance in the graph between two vertices $v$ and $w$):
\begin{equation*}
\begin{split}
& u_A((v,\overrightarrow x),(w,\overrightarrow y))=  4W\cdot \1_{d(v,w)\leq 1}+4W\cdot \1_{x=n(v)}  + val^{(w,\overrightarrow y)}(v) + val^{(v,\overrightarrow x)}(w)+ f_A(v) \\
& u_B((v,\overrightarrow x),(w,\overrightarrow y))=  4W\cdot \1_{d(v,w)\leq 1}+4W\cdot \1_{y=n(w)} 
 + val^{(w,\overrightarrow y)}(v) + val^{(v,\overrightarrow x)}(w)+ f_B(w)
\end{split}  
\end{equation*}
Namely, both players get large reward of $4W$ if they choose adjacent vertices, or the same vertex. Both players get large reward of $4W$ if they report truthfully their own valuations in the neighbourhood of their vertex. Both players get the sum of valuations of the two chosen vertices $v,w$ according to the report of the opponent. In addition Alice gets the (partial) potential of her vertex according to $f_A$, and Bob gets the potential of his vertex according to $f_B$.

\begin{lemma}\label{lem:exact}
The game is an exact potential game.
\end{lemma}
\begin{proof}
We will see that the game can be ``decomposed'' to two exact potential games, and will use this ``decomposition'' to provide a potential function for our game. We will use the following basic properties of potential games. We recall the notation of $(A_1,A_2,u_1,u_2)=(A,u)$ for a two-player game, where each $A_i$ is the action space of player $i$ and $u_i$ is the utility function of player $i$.
\begin{itemize}
\item An \emph{identical interest game} $(A,u)$ is a game in which $u_1=u_2$. An identical interest game is an exact potential game with potential function $\varphi=u_1$.
\item An \emph{opponent independent game} is a game in which the utility of each player $i$ depends only on his own actions: $u_i(a_1,a_2)=u_i(a_i)$ for every $(a_1,a_2)\in A$. Every opponent independent game is an exact potential game where the potential function is simply the sum of the utilities of the players.
\item For every pair of exact potential games $(A,u'),(A,u'')$ with potentials $\varphi',\varphi''$, the game $(A,u'+u'')$ is an exact potential game with potential $\varphi=\varphi'+\varphi''$.
\end{itemize}
Note that our game can be written as a sum of an identical interest game: 
$$u'_A=u'_B= 4W\cdot \1_{d(v,w)\leq 1} + val^{(w,\overrightarrow y)}(v) + val^{(v,\overrightarrow x)}(w)$$
and an opponent independent game:
$$u''_A=4W\cdot \1_{\overrightarrow x=n(v)} + f_A(v), \ u''_B=4W\cdot \1_{\overrightarrow y=n(w)} + f_B(w)$$
Therefore their sum is a potential game with potential:
\begin{align}\label{eq:pot}
\begin{split}
 \phi((v,\overrightarrow x),(y,\overrightarrow w))= & 4W\cdot \1_{d(v,w)\leq 1}+4W\cdot \1_{\overrightarrow x=n(v)}+4W\cdot \1_{\overrightarrow y=n(w)} \\
  & +  val^{(w,\overrightarrow y)}(v) + val^{(v,\overrightarrow x)}(w)+ f_A(v) +f_B(w)
\end{split}  
\end{align}
\end{proof}
%
%
%
%
%
%

\begin{lemma}\label{lem:pne}
The pure Nash equilibria of the game are precisely  $((v,\overrightarrow x),(v,\overrightarrow x'))$ such that $v$ is a local maximum of $f_A+f_B$ and $\overrightarrow x$ and $\overrightarrow x'$ are truth reports of the values of $v$ and its neighbours according to $f_A$ and $f_B$, respectively. 
\end{lemma}

\begin{proof}
Pure Nash equilibria are the local  maxima (with respect to a unilateral deviation) of the potential. It is easy to check that in a local maximum $x$ and $y$ are truth reports, because the gain in a truthful report is $4W$ whereas if the players do not report truthfully they lose this reward. However, Alice can gain at most $val^{(v,\overrightarrow x)}(w)+ f_A(v)\leq 2W$ from misreporting the value, and Bob's loss is similar. Similarly, in a local maximum $v$ and $w$ are neighbours (or the same vertex), because the gain of $4W$ is lost if $v$ and $w$ are not neighbors, in which case Alice's gain from the terms $val^{(w,\overrightarrow y)}(v) + val^{(v,\overrightarrow x)}(w)+ f_A(v)$ is at most $3W$. A similar argument holds for Bob. For a profile of strategies that satisfies the above the potential is equal to (see Equation \eqref{eq:pot}):
\begin{align}\label{eq:tpot}
\phi((v,\overrightarrow x),(w,\overrightarrow y))=12W+f_B(v)+f_A(w) +f_A(v)+ f_A(w)
\end{align}
A profile where $v\neq w$ is not a Nash equilibrium because by the distinctness assumption, $f_A(v)+ f_B(v) \neq f_A(w)+ f_B(w)$, so if $f_A(v)+ f_B(v) < f_A(w)+ f_B(w)$ Alice can deviate to $(w,\overrightarrow y)$ and increase the potential; Otherwise Bob can deviate to $(v,\overrightarrow x)$ and increase the potential. Finally, a profile $((v,\overrightarrow x),(v,\overrightarrow x'))$ with truth reporting is clearly a Nash equilibrium if it is a local maximum of $f_A+f_B$. If $v$ is not a local maximum of $f_A+f_B$, then Alice will increase the potential (given in Equation \eqref{eq:tpot}) if she deviates to the action $(w,n(w))$ where $w$ is a neighbour of $v$ with $f_A(w)+f_B(w)>f_A(v)+f_B(v)$.
\end{proof}

Lemmas \ref{lem:exact} and \ref{lem:pne} complete the proof of the theorem.

\section{Proof of Theorem \ref{theo:n-pot}}\label{sec:pr-n}

The proof of Theorem \ref{theo:n-pot} is done in two steps. First, we show a $2^{\Omega(\sqrt[3]{n})}$ bound. This is the significant part, in terms of the deduced result and also in terms of the techniques. Thereafter, in Section \ref{sec:2n} we improve the bound to $2^{\Omega(n)}$ building upon the arguments of this Section.

We start with proving the $2^{\Omega(\sqrt[3]{n})}$ bound. Our starting point is the proof of the hardness of $2$-player $n$-actions exact potential games (Theorem \ref{theo:2pot}). However, since we consider $n$-player binary-action games, it is convenient to reduce the problem $\SumLS(\Hyp_n)$ (local search on the $n$-th hypercube). We will get an exact potential game with $\Theta(n^3)$ players, where each player has only two actions.

\paragraph{A naive approach and an obstacle. } The simplest idea that comes to mind is to consider a \emph{group} of $n$-players who will choose $v\in \Hyp_n$, and a \emph{group} of $(n+1) \lceil\log W \rceil$ players who will report the valuation vector $\overrightarrow x$ of the vertex itself and its $n$ neighbours, and similarly for Bob. We would like to set the group of Alice's players an \emph{identical utility} that is similar to the utility of Alice in the two-player game. An obstacle that arises with this approach is that if the groups of Alice's and Bob's players are playing two adjacent vertices $v,w\in \Hyp_n$ with truthful valuations, none of them will want to switch to the opponent's vertex, even if at the adjacent vertex the sum of $f_A+f_B$ is higher. This follows from the fact if $(v,\overrightarrow x)$ is a truthful valuation, then $(w,\overrightarrow x)$ is not necessarily a truthful valuation (because the relevant vertices and their order is different with respect to $v$ and with respect to $w$). Thus, players in Alice's group will gain the difference in the potentials (at most $3W$) but lose $4W$ because now the group report is not truthful. Note that the same obstacle does not arise in the two-player case. In the two-player case Alice could change the vertex $v$ and the report $\overrightarrow x$ \emph{simultaneously}. In the $n$-player case we consider unilateral deviations that correspond to changes of single bits and thus such simultaneous deviations are impossible.

\paragraph{The solution to the obstacle. } To resolve the above problematic issue, we modify the form of the report $\overrightarrow x$ in the game.
\begin{itemize}
\item Instead of reporting the values in the ball of radius 1 around $v$ (i.e., the neighbors of $v$), each player reports the values in the ball of radius 2 around $v$. In the hypercube, this means that the report consists of $m=1+n+\frac{n(n-1)}{2}$ valuations.
\item Instead of reporting the values in a fixed order (namely $(v,n_1(v),...,n_4(v))$), the players jointly report pairs, where each pair consists of an index of a vertex $v$ and $f_A(v)$ (or $f_B(v)$).
\end{itemize}

\paragraph{The construction. }More formally, for Alice, we have a group of $n$ players with binary actions who jointly choose the vertex $v\in \{0,1\}^n$. In other words, the action of the $i$'th player in the group corresponds to the $i$'th bit in the index of the vertex. We have a group of $mn$ players with binary actions who jointly choose a list of $m$ vertices $\overrightarrow{xv}= (xv_1,...,xv_m)\in (\{0,1\}^n)^m$. Finally, we have a group of $m b:=m\lceil \log W \rceil$ players with binary actions who jointly choose a list of $m$ valuations $\overrightarrow{xf}=(xf_1,...,xf_m)\in (\{0,1\}^b)^m$. We denote $\overrightarrow{x}=(\overrightarrow{xv},\overrightarrow{xf})$. 
Similarly to the two-player case, a report $\overrightarrow{x}=(\overrightarrow{xv},\overrightarrow{xf})$ defines a valuation function over all vertices. 
For a list $\overrightarrow{xv}$ we denote $I_{\min}(\overrightarrow{xv}):=\{i\in [m]: xv_i \neq xv_j \text{ for all } j<i\}$ the set of indices with \emph{first} appearance of a vertex. The valuation is defined by  
\begin{align*}
val^{\overrightarrow{x}}(w)=\begin{cases}
val(xf_i) &\text{if } w=xv_i \text{ for } i\in I_{\min}(\overrightarrow{xv}); \\
0 &\text{otherwise.}
 \end{cases}
\end{align*}
where $val(\cdot)\in [W]$ denotes the numerical value of the binary string. Note that in case of multiple appearances of $w$ in the list we choose the value at the first appearance.
Similarly for Bob, we have three groups who jointly choose $w$, $\overrightarrow{yw}$, and $\overrightarrow{yf}$. The report $\overrightarrow{y}$ defines a valuation function $val^{\overrightarrow{y}}$ over all vertices. Note that the total number of players in the game is $2(n+m(n+b))=O(n^3)$.

Before we present the actual utilities we informally describe the prioritization according to which we set the utilities. In the two-player case there were only two levels of prioritization: the top level priority included the distance $d(v,w)$ (the $\1_{d(v,w)\leq 1}$ term in the utility functions) and the truthfulness of the report (the $\1_{x=n(v)}$ term in the utility functions). The bottom level priority included the remaining potential related terms ($val^{(w,\overrightarrow y)}(v), val^{(v,\overrightarrow x)}(w), f_A(v)$). More formally by \emph{prioritization} we mean that improving the higher priority term by $1$ should increase the utility \emph{irrespective} of how the lower priority terms change. Indeed the multiplier $4W$ was set in such a way. In the current construction, the prioritization levels are more involved, and we sketch them here from the highest priority to the lowest.
\begin{enumerate}
\item The distance $d(v,w)$.
\item The list $\overrightarrow{xv}$ should contain $v$ and its neighbours.
\item The valuations $\overrightarrow{xf}$ should be correct for $v$ and its neighbours.
\item The potential related terms (the core of the proof).
\item The list $\overrightarrow{xv}$ should contain the vertices within a distance 2 from $v$.
\item The valuations $\overrightarrow{xf}$ should be correct for vertices within a distance 2 from $v$.
\end{enumerate}

Now we describe what is the analogue of each one of these priorities in the $n$-player case. Hereafter, $d(\cdot,\cdot)$ will denote the \emph{hamming distance} (in the corresponding dimension). We denote by $B_r(v)$ the ball of radius $r$ around $v$ with respect to the hamming distance.

\begin{enumerate}
\item $\1_{d(v,w)\leq 1}$ is translated to $-d(v,w)\cdot \1_{d(v,w)\geq 2}$. Namely the loss is 0 in case the players choose the same vertex or adjacent vertices. Otherwise the loss increases with the distance.
\item Given $v$, we denote by $N_1(v):=\{(v_1,...,v_m): \{v_1,...,v_m\} \supset B_1(v)\}\subset \{0,1\}^{mn}$. Namely, $N_1(v)$ specifies $v$ and its neighbours. At the second priority we have $-d(\overrightarrow{xv},N_1(v))$.
\item Given $v$ and $\overrightarrow{xv}$, for an index $i\in I_{\min}(\overrightarrow{xv})$ such that $xv_i\in B_1(v)$ we have at the third priority the term $-d(xf,bin(f_A(xv_i))$ when we recall that $bin(z)\in \{0,1\}^b$ represents the binary representation of the potential value $z\in [W]$. Note that this definition takes into account only the \emph{first} appearance of every neighbour, which is consistent with the definition of $val^{\overrightarrow{x}}$. For other indices $i\in [m]$ the term will be identical but it will appear at the lowest sixth priority.
\item The profile $(v,\overrightarrow{x}),(w,\overrightarrow{y})$ defines a natural analogue of the two-player potential terms: $val^{\overrightarrow y}(v), val^{\overrightarrow y}(w), f_A(v), f_B(w)$. These terms are at the forth priority.
\item Given $v$, we denote by $N_2(v):=\{(v_1,...,v_m): \{v_1,...,v_m\} = B_2(v)\}\subset \{0,1\}^{mn}$ the lists that include precisely the set of all vertices within a radius 2 from $v$. At the fifth priority we have $-d(\overrightarrow{xv},N_2(v))$.
\item Finally, similarly to item 3, given $v$ and $\overrightarrow{xv}$, for every index $i\in [m]$ we have at the sixth priority the term $-d(xf,bin(f_A(xv_i))$.
\end{enumerate} 

Now we are ready to define the utilities. As was mentioned above all the players in Alice's groups have identical utilities which is equal to:

\begin{align*}
u^A_i(v,\overrightarrow{x},w,\overrightarrow{y})= 
& - k_1 \cdot d(v,w)\1_{d(v,w)\geq 2} \\
& - k_2 \cdot d(\overrightarrow{xv},N_1(v)) \\
& - k_3 \cdot \sum_{i\in I_{\min}(\overrightarrow{xv}) \text{ s.t. } xv_i\in B_1(v)} d(xf,bin(f_A(xv_i)) \\
& + k_4 [val^{\overrightarrow y}(v) + val^{\overrightarrow x}(w) + f_A(v)] \\
& - k_5 \cdot d(\overrightarrow{xv},N_2(v)) \\
& - k_6 \cdot \sum_{i\in [m]} d(xf,bin(f_A(xv_i)),
\end{align*}
when we set $k_1,...,k_6$ as follows. We set $k_6=1$. Now we set $k_5$ to be greater than the maximal difference of sixth priority terms, e.g., $k_5=2n^2 b>mb$. Now we set $k_4$ to be the greater than the maximal total difference of sixth and fifth priority terms, e.g., $k_4=2n^3 b>mb+k_5(nm)$. Similarly we may proceed with $k_3 = 8W n^3 b$, $k_2 = 8Wn^5 b^2$, and $k_1=8Wn^8 b^2$.

Similarly we define each member in Bob's group to have the following identical utility function:
\begin{align*}
u^B_i(v,\overrightarrow{x},w,\overrightarrow{y})= 
& - k_1 \cdot d(v,w)\1_{d(v,w)\geq 2} \\
& - k_2 \cdot d(\overrightarrow{yw},N_1(w)) \\ 
& - k_3 \cdot \sum_{i\in I_{\min}(\overrightarrow{yw}) \text{ s.t. } yw_i\in B_1(w)} d(yf,bin(f_A(yw_i)) \\
& + k_4 [val^{\overrightarrow y}(v) + val^{\overrightarrow x}(w) + f_B(w)] \\
& - k_5 \cdot d(\overrightarrow{wy},N_2(w))\\ 
& - k_6 \cdot \sum_{i\in [m]} d(yf,bin(f_A(yw_i)).
\end{align*}

\begin{lemma}\label{lem:potential}
The defined $(2n+2m(n+b))$-player binary action game is an exact potential game. 
\end{lemma}
\begin{proof}
If we view the game as a \emph{two}-player game where Alice chooses $(s,\hat{x})$ and Bob chooses $(r,\hat{y})$ the game is an exact potential game by similar arguments to those in Lemma \ref{lem:exact}. Namely it is the sum of two games where one is identical interest game and the other is opponent independent game. The potential function of the game is given by: 
\begin{align*}
\phi(v,\overrightarrow{x},
& w,\overrightarrow{y})= - k_1  \cdot d(v,w)\1_{d(v,w)\geq 2} \\
& - k_2 [d(\overrightarrow{xv},N_1(v))+d(\overrightarrow{yw},N_1(w))] \\
& - k_3 [\sum_{i\in I_{\min}(\overrightarrow{xv}) \text{ s.t. } xv_i\in B_1(v)} d(xf,bin(f_A(xv_i)) + \sum_{i\in I_{\min}(\overrightarrow{yw}) \text{ s.t. } yw_i\in B_1(w)} d(yf,bin(f_A(yw_i))] \\
& + k_4 [val^{\overrightarrow y}(v) + val^{\overrightarrow x}(w) + f_A(v)+f_B(w)] \\
& - k_5 [d(\overrightarrow{xv},N_2(v))+d(\overrightarrow{yw},N_1(w))] \\ 
& - k_6 [\sum_{i\in [m]} d(xf,bin(f_A(xv_i))+\sum_{i\in [m]} d(yf,bin(f_B(yw_i))]
\end{align*}
Note that by replacing Alice (Bob) by a group of $n+m(n+b)$ players all with the same utility we only reduced the set of possible unilateral deviations. For each one of these unilateral deviation by the two-player result the change is the utility is equal to the change in the potential. 
\end{proof}

\begin{lemma}\label{lem:npne}
Every pure Nash equilibrium of the defined $(2n+2m(n+b))$-player binary action game is of the form $(v,\overrightarrow{x},v,\overrightarrow{y})$ where $v$ is a local maximum of $f_A+f_B$ over the hypercube.
\end{lemma}

\begin{proof}
The proof proceeds by narrowing the set of equilibria candidates according to the prioritization levels, with a twist at the fourth priority level.

First, in every equilibrium $d(v,w)\leq 1$ because otherwise there exists a player in Alice's $v$ group who can switch his strategy and decrease the distance by 1. Such a switch increases the first term in the utility of the group by $k_1$. By the choice of $k_1$, any change in the other terms of utilities is smaller.

Second, in every equilibrium $\overrightarrow{xv}\in N_1(v)$, because otherwise there exists a player in Alice's $\overrightarrow{xv}$ group who can switch his strategy and decrease the distance by 1. Such a switch does not effect the first term of the utility, and it increases the second term by $k_2$. By the choice of $k_2$, any change in the other terms of utilities is smaller. Similarly for Bob we have $\overrightarrow{yw}\in N_1(w)$.

Third, in every equilibrium for every $i\in I_{\min}(\overrightarrow{xv})$ such that $xv_i\in B_1(v)$ we have $xf_i=bin(f_A(xv_i))$. Simply speaking, all first appearances of elements in $B_1(v)$ (which indeed appear by the argument regarding the second priority level) have correct valuation. If it wasn't so, then there exists a player in Alice's $\overrightarrow{xf_i}$ group who can switch his strategy and decrease the distance by 1. Such a switch does not affect the first two terms of the utility, and it increases the third term by $k_3$. By the choice of $k_3$, any change in the other terms of utilities is smaller. Similarly for Bob, all first appearances of elements in $B_1(w)$ have correct valuation.

Now we jump to the fifth and the sixth priority levels. Given that $v,w$ are neighbours (or the same vertex) and their values already appear in the report $\overrightarrow{x}$ the terms of the utility in the fourth priority level are not affected by the vertices $xv_i$ such that $i\notin I_{\min}(\overrightarrow{xv})$ or $xv_i \notin B_1(v)$. Therefore, we can deduce that necessarily in equilibrium we have $\overrightarrow{xv}\in N_2(v)$ because otherwise some player in the $\overrightarrow{xv}$ group can decrease the distance by 1 without affecting any of the first four terms, and increase the fifth term by $k_5$. Any change in the last terms is smaller.
Similarly we can argue for the sixth priority level, that the values of $xf_i$ for the corresponding indices do not affect any other term.
From these arguments it follows that in any equilibrium both Alice (and Bob) report a list $\overrightarrow{xv}$ ($\overrightarrow{yw}$) that contains exactly all the vertices in the ball of radius 2 around $v$ ($w$), moreover all valuations of all these vertices are correct.

Now we go back to the fourth priority. Assume by way of contradiction that $v\neq w$. Similarly to the two-player case, the fourth term in the \emph{potential function} of the game is $val^{\overrightarrow{x}}(w)+val^{\overrightarrow{y}}(v)+f_A(v)+f_B(w)$, which under all the above restrictions of equilibria is equal to $f_A(w)+f_B(v)+f_A(v)+f_B(w)$. 

Assume by way of contradiction that $v\neq w$, then we may assume w.l.o.g. that $f_A(w)+f_B(w)\geq f_A(v)+f_B(v)+1$ (we recall that we may assume that the sum defers at adjacent vertices and has integer values), then there exists a player in Alice's $v$ group who can switch his bit and turn the vertex $v$ into $w$. Let us examine the effect of this change on the potential. The first priority level term remains 0. The key observation is that the second and third priority level terms also remain 0. Note that the list $\overrightarrow{xv}$ includes all the vertices within radius 2 from $v$, and in particular all the vertices within radius 1 from $w$. Similarly the valuations $\overrightarrow{xf}$ of these vertices remain correct. Therefore the potential increases by at least $k_4$ in the first four terms, and any change in the fifth and sixth terms is smaller.

Finally for the case of $v=w$ where $v$ is not a local maximum we apply very similar arguments: There exists a player in Alice's group who can increase the potential of the game by $k_3$ and change only the fifth and sixth terms of the potential.  
 
\end{proof}
Lemmas \ref{lem:potential}, and \ref{lem:npne} complete the proof of the $2^{\Omega(\sqrt[3]{n})}$ bound.

\subsection{Proof of Theorem \ref{theo:n-pot}: Improving the Bound to $2^{\Omega(n)}$}\label{sec:2n}
The presented above reduction has $\Theta(n^3)$ players, which yields a lower bound of $2^{\Omega(\sqrt[3]{n})}$ on the problem of finding a pure Nash equilibrium. Here we modify the reduction to have $\Theta(n)$ players, which implies a lower bound of $2^{\Omega(n)}$ on the problem of finding a pure Nash equilibrium. The idea is to reduce the unnecessary ``wasting" of players in the reduction. In the presented reduction Alice reports to Bob the valuations of \emph{all} vertices within radius 2 around $v$ (there are $\Theta(n^2)$ such vertices). However, the arguments of the proof of Theorem \ref{theo:grid} can be modified to show the existence of hard instances over the hypercube where for most of the neighbours within radius 2 from $v$, Alice and Bob \emph{know} the valuations of each other over these vertices. In fact, for these hard instances there exist only a \emph{constant} number of neighbours for which Alice does not know Bob's valuation, and Bob does not know Alice's. In the modified reduction, Alice's group will report only the valuation of the unknown vertices, which will require only $O(n)$ players for her group.

We start with a modification of Lemma \ref{lem:hyp}, which embeds the constant degree graph $G$ in $\Hyp_n$. We present an embedding of $G$ in $\Hyp_n$ with the additional property that every ball of radius 2 in $\Hyp_n$ contains at most \emph{constant} number of vertices of the embedding's image. Formally, given an embedding $(\varphi,\chi)$ where $\varphi:V_G \ra \{0,1\}^n$, $\chi:E_G \ra P(\Hyp_n)$, we denote the \emph{image of the embedding} by $Im(G)=\{w\in \{0,1\}^n: w\in \varphi(V_G) \cup \chi(E_G)\}$.

\begin{lemma}\label{lem:hyp-const}
Let $G$ be the graph with $N$ vertices that is defined in Section \ref{sec:veto} (the constant degree graph for which Theorem \ref{theo:opt}(\ref{theo:opt-bounded}) holds). 
The graph $G$ can be VIED-embedded in $\Hyp_n$ for $n=O(\log N)$, such that for every $w\in\{0,1\}^n$ we have\footnote{More concretely, $n=3\log N+333$ and $|B_2(w)\cap Im(G)|\leq 73$.} $|B_2(w)\cap Im(G)|=O(1)$.
\end{lemma}

\begin{proof}
We ``sparse" the embedding of Lemma \ref{lem:hyp} to reach a situation where every pair of independent edges are embedded to paths that are within a distance of at least 3 one from the other. This can be done, for instance, by embedding $G$ in a hypercube of dimension $n=3(\log N+111)$ rather than dimension $n'=\log N +111$, when we replace every vertex in $\Hyp_{n'}$ by three copies of itself. Such a change multiplies the hamming distance by a factor of 3. For such an embedding the maximal number of vertices of $Im(G)$ in a ball or radius 2 is obtained at a vertex $w\in \phi(V_G)$ and is equal to $1+2\cdot 32$; the vertex and two vertices of every one of the 36 embedded edges. 
\end{proof}

We proceed with a short presentation of the arguments that prove Theorem \ref{theo:opt}(\ref{theo:opt-hypercube}) from Theorem \ref{theo:opt}(\ref{theo:opt-bounded}), followed by a Corollary that will be essential in our reduction. The arguments below are very similar, but yet slightly defer from the proof that is presented in Section  \ref{sec:main-pr}.

\paragraph{Proof of Theorem \ref{theo:opt}(\ref{theo:opt-hypercube}) from Theorem \ref{theo:opt}(\ref{theo:opt-bounded}). } We reduce $\SumLS(G)$ to $\SumLS(\Hyp_n)$ using the VIED embedding $(\varphi,\chi)$ of Lemma \ref{lem:hyp-const}. Let $Im(G)\subset \Hyp_n$ be the image of the embedding and let $w^*\in Im(G)$ be some fixed vertex. Given an instance $(f_A,f_B)$ of $\SumLS(G)$ we define an instance $(f'_A,f'_B)$ of $\SumLS(\Hyp_n)$ by
\begin{align*}
f'_A(w)=
\begin{cases}
f_A(v) 									  &\text{if } w=\varphi(v)\in \varphi(V_G) \\
\frac{k}{l} f_A(v) + \frac{l-k}{l}f_A(v') &\text{if } w\in \chi(\{v,v'\})\subset \chi(E_G) \\
-d(w,w^*)								  &\text{otherwise.}
\end{cases}\\
f'_B(w)=
\begin{cases}
f_B(v) 									  &\text{if } w=\varphi(v)\in \varphi(V_G) \\
\frac{k}{l} f_B(v) + \frac{l-k}{l}f_B(v') &\text{if } w\in \chi(\{v,v'\})\subset \chi(E_G) \\
-d(w,w^*)										  &\text{otherwise.}
\end{cases}
\end{align*}
where in the case $w\in \chi(\{v,v'\})$ we assume that $w$ is the $k$'th element in the path $\chi(\{v,v'\})$ and $l$ is the total length of this path. Simply speaking, we set the functions $f'_A,f'_B$ to have the values $f_A(v)$ on the embedded vertices $\varphi(v)$. On intermediate vertices along a path that embeds an edge we set the value to be a weighted average of the two extreme valuations. For vertices out of $Im(G)$ we set $f'_A(w)=f'_B(w)$ to be a negative constant that does not depend on the instance $(f_A,f_B)$.

It can be easily checked that the local maxima of $f'_A+f'_B$ over $\Hyp_n$ are precisely $\{\varphi(v): v \text{ is a local maximum of } f_A+f_B \text{ over } G\}$.
\begin{corollary}\label{cor:img-fixed}
Finding local maximum in $\Hyp_n$ with the promise of $f'_A(w)=f'_B(w)=-d(w,w^*)$ for all $w\notin Im(G)$, requires $2^{\Omega(n)}$ communication. 
\end{corollary}

Now we construct a potential game with $\Theta(n)$-players that solves the \emph{promise} $\SumLS$ problem of Corollary \ref{cor:img-fixed}.

We mimic the arguments of the previous $2^{\sqrt[3]{n}}$ bound with one change: the reports $\overrightarrow{x}$ and $\overrightarrow{y}$ are done on vertices in $B_2(v)\cap Im(G)$ rather than $B_2(v)$. By Lemma \ref{lem:hyp-const} it is sufficient to report $73$ vertices (rather than $\Theta(n^2)$). A report consists of $\overrightarrow{x}=(\overrightarrow{xv},\overrightarrow{xf})$ where $\overrightarrow{xv}$ is a 73-tuple of vertices, and $\overrightarrow{xf}$ is a 73-tuple of valuations. The valuation function $val^{\overrightarrow{x}}(w)$ is modified to be $val^{\overrightarrow{x}}(w) = val(xf_i)$ if  $w=xv_i$ for  $i\in I_{\min}(\overrightarrow{xv})$; Otherwise, if $w\notin Im(G)$ we set $val^{\overrightarrow{x}}(w)=-d(w,w^*)$; Otherwise, we set $val^{\overrightarrow{x}}(w)=0$.
Similarly for Bob. 

We also modify the definition of the neighbour vertices of $v\in \Hyp_n$: 
\begin{align*}
N_1(v)&:=\{(v_1,...,v_73): \{v_1,...,v_73\} \supset B_1(v)\cap Im(G)\}\subset \{0,1\}^{73n} \\
N_2(v)&:=\{(v_1,...,v_9): \{v_1,...,v_73\} \supset B_2(v) \cap Im(G)\}\subset \{0,1\}^{73n}.
\end{align*}
Note that by the Lemma \ref{lem:hyp-const} $N_1(v),N_2(v)\neq \emptyset$ for all vertices $v$.

From here, we apply similar arguments to those in Section \ref{sec:npot-pr} to prove a reduction from the local search promise problem of Corollary \ref{cor:img-fixed} to pure Nash equilibrium in potential games. The only additional argument that is needed is that $val^{\overrightarrow{x}},val^{\overrightarrow{y}}$ have the correct valuation for all vertices $w\notin Im(G)$ (in particular those within radius 2).

%% file: id-ord.tex
\section{Identifying Ordinal Potential Games}\label{ap:ident-ord}

We will prove two results, one for two-player $N$-action games and one for $n$-player $2$ action games. In both we use the following two-player two-action game for $x,y\in \{0,2\}$:

\begin{center}
\begin{tabular}{|l|l|}
\hline
$2,1$ & $1,2$ \\ \hline
$1,y$ & $x,1$ \\ \hline
\end{tabular}
\end{center}
\noindent
This game has a better-reply cycle if and only if $x=y=2$.

\begin{proposition}\label{pro:2-ord}
Recognizing whether a two-player $N$-action game is an ordinal potential game requires $\poly(N)$ bits of communication, even for randomized protocols.
\end{proposition}
\begin{proof}
Denote by $u'$ the two-player $2N\times 2N$ table that contains $N\times N$ copies of this game with the parameters $(x_{i,j},y_{i,j})_{i,j\in [N]}$. We denote by $u''$ the two-player $2N\times 2N$ game with the payoffs $u''(a,b)=(3\lceil \frac{a}{2}\rceil,3\lceil \frac{b}{2}\rceil)$. And we denote $u=u'+u''$. The game $u$ has a better-reply cycle if and only if there exist $i,j\in [N]$ such that $x_{i,j}=y_{i,j}=2$. Indeed if $x_{i,j}=y_{i,j}=2$, since we have added a constant payoff of $3i$ to player $1$ ($3j$ to player 2) to the $(i,j)$ copy of the game, the better reply cycle remains a better reply cycle in $u$. If $(x_{i,j},y_{i,j})\neq (2,2)$ for all $i,j$ then we have no better reply cycle within the copies of the $2\times 2$ games, and we have no better reply cycles across the $2\times 2$ games because at least one player has dominant strategy. Therefore the determination of ordinal potential property is as hard as disjointness, which requires $\poly(N)$ communication, even with randomized communication.
\end{proof}

\begin{proposition}\label{pro:n-ord}
Recognizing whether an $n$-player $2$-action game is an ordinal potential game requires $2^{\Omega(n)}$ bits of communication, even for randomized protocols.
\end{proposition}
\begin{proof}
Consider an $(n+2)$-player game where for each profile $a\in \{0,1\}^n$ the last two players are playing the above $2\times 2$ game with parameters $x_a,y_a$. For the first $n$ players we set the utilities such that $1$ is dominant strategy (e.g., $u_i(a_i,a_{-i})=a_i$ for $i\in [n]$). Similarly to the previous arguments, the game contains a better reply cycle if and only if $x_a=y_a=2$ for some $a\in \{0,1\}^n$. Again, we obtain a reduction to disjointness. 
\end{proof}

%% file: CC-LS1.bbl
\begin{thebibliography}{30}
\providecommand{\natexlab}[1]{#1}
\providecommand{\url}[1]{\texttt{#1}}
\expandafter\ifx\csname urlstyle\endcsname\relax
  \providecommand{\doi}[1]{doi: #1}\else
  \providecommand{\doi}{doi: \begingroup \urlstyle{rm}\Url}\fi

\bibitem[Aaronson(2006)]{Aar}
S.~Aaronson.
\newblock Lower bounds for local search by quantum arguments.
\newblock \emph{SIAM Journal on Computing}, 35\penalty0 (4):\penalty0 804--824,
  2006.

\bibitem[Aldous(1983)]{Ald}
D.~Aldous.
\newblock Minimization algorithms and random walk on the d-cube.
\newblock \emph{The Annals of Probability}, pages 403--413, 1983.

\bibitem[Anshu et~al.(2017)Anshu, Goud, Jain, Kundu, and Mukhopadhyay]{AGJKM}
A.~Anshu, N.~B. Goud, R.~Jain, S.~Kundu, and P.~Mukhopadhyay.
\newblock Lifting randomized query complexity to randomized communication
  complexity.
\newblock \emph{arXiv preprint arXiv:1703.07521}, 2017.

\bibitem[Babaioff et~al.(2018)Babaioff, Dobzinski, and Oren]{BDO18}
M.~Babaioff, S.~Dobzinski, and S.~Oren.
\newblock Combiantorial auctions with endowment effect.
\newblock 2018.

\bibitem[Babichenko and Rubinstein(2016)]{BR}
Y.~Babichenko and A.~Rubinstein.
\newblock Communication complexity of approximate {N}ash equilibria.
\newblock \emph{arXiv preprint arXiv:1608.06580}, 2016.

\bibitem[Beame et~al.(1998)Beame, Cook, Edmonds, Impagliazzo, and
  Pitassi]{BCEI}
P.~Beame, S.~Cook, J.~Edmonds, R.~Impagliazzo, and T.~Pitassi.
\newblock The relative complexity of np search problems.
\newblock \emph{Journal of Computer and System Sciences}, 57\penalty0
  (1):\penalty0 3--19, 1998.

\bibitem[Christodoulou et~al.(2008)Christodoulou, Kov{\'a}cs, and
  Schapira]{christodoulou2008bayesian}
G.~Christodoulou, A.~Kov{\'a}cs, and M.~Schapira.
\newblock Bayesian combinatorial auctions.
\newblock In \emph{International Colloquium on Automata, Languages, and
  Programming}, pages 820--832. Springer, 2008.

\bibitem[Conitzer and Sandholm(2004)]{CS}
V.~Conitzer and T.~Sandholm.
\newblock Communication complexity as a lower bound for learning in games.
\newblock In \emph{Proceedings of the twenty-first international conference on
  Machine learning}, page~24. ACM, 2004.

\bibitem[Daskalakis et~al.(2009)Daskalakis, Goldberg, and Papadimitriou]{DGP}
C.~Daskalakis, P.~W. Goldberg, and C.~H. Papadimitriou.
\newblock The complexity of computing a nash equilibrium.
\newblock \emph{SIAM Journal on Computing}, 39\penalty0 (1):\penalty0 195--259,
  2009.

\bibitem[G{\"o}{\"o}s and Pitassi(2014)]{GP14}
M.~G{\"o}{\"o}s and T.~Pitassi.
\newblock Communication lower bounds via critical block sensitivity.
\newblock In \emph{Proceedings of the forty-sixth annual ACM symposium on
  Theory of computing}, pages 847--856. ACM, 2014.

\bibitem[G{\"o}{\"o}s and Rubinstein(2018)]{GoosRub}
M.~G{\"o}{\"o}s and A.~Rubinstein.
\newblock Near-optimal communication lower bounds for approximate nash
  equilibria.
\newblock \emph{Proceedings of 59th Annual Symposium on Foundations of Computer
  Science, 2018.}, 2018.

\bibitem[G{\"o}{\"o}s et~al.(2015)G{\"o}{\"o}s, Pitassi, and Watson]{GPW15}
M.~G{\"o}{\"o}s, T.~Pitassi, and T.~Watson.
\newblock Deterministic communication vs. partition number.
\newblock In \emph{Foundations of Computer Science (FOCS), 2015 IEEE 56th
  Annual Symposium on}, pages 1077--1088. IEEE, 2015.

\bibitem[G{\"o}{\"o}s et~al.(2017)G{\"o}{\"o}s, Pitassi, and Watson]{GPW}
M.~G{\"o}{\"o}s, T.~Pitassi, and T.~Watson.
\newblock Query-to-communication lifting for {BPP}.
\newblock \emph{arXiv preprint arXiv:1703.07666}, 2017.

\bibitem[Hart and Mansour(2010)]{HMan}
S.~Hart and Y.~Mansour.
\newblock How long to equilibrium? {T}he communication complexity of uncoupled
  equilibrium procedures.
\newblock \emph{Games and Economic Behavior}, 69\penalty0 (1):\penalty0
  107--126, 2010.

\bibitem[Hart and Mas-Colell(2003)]{HMas}
S.~Hart and A.~Mas-Colell.
\newblock Uncoupled dynamics do not lead to {N}ash equilibrium.
\newblock \emph{American Economic Review}, 93\penalty0 (5):\penalty0
  1830--1836, 2003.

\bibitem[Hub{\'a}cek et~al.(2017)Hub{\'a}cek, Naor, and
  Yogev]{hubacek2017journey}
P.~Hub{\'a}cek, M.~Naor, and E.~Yogev.
\newblock The journey from np to tfnp hardness.
\newblock In \emph{LIPIcs-Leibniz International Proceedings in Informatics},
  volume~67. Schloss Dagstuhl-Leibniz-Zentrum fuer Informatik, 2017.

\bibitem[Johnson et~al.(1988)Johnson, Papadimitriou, and Yannakakis]{JPY}
D.~S. Johnson, C.~H. Papadimitriou, and M.~Yannakakis.
\newblock How easy is local search?
\newblock \emph{Journal of computer and system sciences}, 37\penalty0
  (1):\penalty0 79--100, 1988.

\bibitem[Karchmer and Wigderson(1990)]{KW}
M.~Karchmer and A.~Wigderson.
\newblock Monotone circuits for connectivity require super-logarithmic depth.
\newblock \emph{SIAM Journal on Discrete Mathematics}, 3\penalty0 (2):\penalty0
  255--265, 1990.

\bibitem[Karchmer et~al.(1995)Karchmer, Raz, and Wigderson]{KRW}
M.~Karchmer, R.~Raz, and A.~Wigderson.
\newblock Super-logarithmic depth lower bounds via the direct sum in
  communication complexity.
\newblock \emph{Computational Complexity}, 5\penalty0 (3-4):\penalty0 191--204,
  1995.

\bibitem[Megiddo and Papadimitriou(1991)]{MP}
N.~Megiddo and C.~H. Papadimitriou.
\newblock On total functions, existence theorems and computational complexity.
\newblock \emph{Theoretical Computer Science}, 81\penalty0 (2):\penalty0
  317--324, 1991.

\bibitem[Monderer and Shapley(1996)]{MS}
D.~Monderer and L.~S. Shapley.
\newblock Potential games.
\newblock \emph{Games and Economic Behavior}, 14\penalty0 (1):\penalty0
  124--143, 1996.

\bibitem[Nisan(1993)]{nisan1993communication}
N.~Nisan.
\newblock The communication complexity of threshold gates.
\newblock \emph{Combinatorics, Paul Erdos is Eighty}, 1:\penalty0 301--315,
  1993.

\bibitem[Nisan(2009{\natexlab{a}})]{Noam-blog}
N.~Nisan.
\newblock Communication complexity of reaching equilibrium.
\newblock In \emph{Turing's Invisible Hand blog,
  https://agtb.wordpress.com/2009/08/18/communication-complexity-of-reaching-equilibrium},
  2009{\natexlab{a}}.
\newblock URL
  \url{https://agtb.wordpress.com/2009/08/18/communication-complexity-of-reaching-equilibrium/}.

\bibitem[Nisan(2009{\natexlab{b}})]{Noam-blog-2}
N.~Nisan.
\newblock The computational complexity of pure nash.
\newblock In \emph{Turing's Invisible Hand blog,
  https://agtb.wordpress.com/2009/11/19/the-computational-complexity-of-pure-nash},
  2009{\natexlab{b}}.
\newblock URL
  \url{https://agtb.wordpress.com/2009/11/19/the-computational-complexity-of-pure-nash/}.

\bibitem[Nisan and Segal(2006)]{nisan2006communication}
N.~Nisan and I.~Segal.
\newblock The communication requirements of efficient allocations and
  supporting prices.
\newblock \emph{Journal of Economic Theory}, 129\penalty0 (1):\penalty0
  192--224, 2006.

\bibitem[Papadimitriou et~al.(1990)Papadimitriou, Schaeffer, and
  Yannakakis]{PSY}
C.~H. Papadimitriou, A.~A. Schaeffer, and M.~Yannakakis.
\newblock On the complexity of local search.
\newblock In \emph{Proceedings of the twenty-second annual ACM symposium on
  Theory of computing}, pages 438--445. ACM, 1990.

\bibitem[Raz and McKenzie(1997)]{RM}
R.~Raz and P.~McKenzie.
\newblock Separation of the monotone nc hierarchy.
\newblock In \emph{Proceedings of 38th Annual Symposium on Foundations of
  Computer Science, 1997.}, pages 234--243. IEEE, 1997.

\bibitem[Santha and Szegedy(2004)]{SS}
M.~Santha and M.~Szegedy.
\newblock Quantum and classical query complexities of local search are
  polynomially related.
\newblock In \emph{Proceedings of the thirty-sixth annual ACM symposium on
  Theory of computing}, pages 494--501. ACM, 2004.

\bibitem[Sun and Yao(2009)]{SY}
X.~Sun and A.~C.-C. Yao.
\newblock On the quantum query complexity of local search in two and three
  dimensions.
\newblock \emph{Algorithmica}, 55\penalty0 (3):\penalty0 576--600, 2009.

\bibitem[Viola(2015)]{viola2015communication}
E.~Viola.
\newblock The communication complexity of addition.
\newblock \emph{Combinatorica}, 35\penalty0 (6):\penalty0 703--747, 2015.

\end{thebibliography}
